\documentclass[12pt]{amsart}

\usepackage{amsmath,amssymb,amsthm,textcomp}
\usepackage{amsfonts,graphicx}
\usepackage[mathscr]{eucal}
\usepackage{color}

\theoremstyle{definition}

\numberwithin{equation}{section}

\newcommand{\ncom}{\newcommand}

\ncom{\beq}{\begin{equation}}
\ncom{\eeq}{\end{equation}}
\ncom{\bea}{\begin{eqnarray*}}
\ncom{\eea}{\end{eqnarray*}}
\ncom{\beqa}{\begin{eqnarray}}
\ncom{\eeqa}{\end{eqnarray}}
\ncom{\nno}{\nonumber}
\ncom{\non}{\nonumber}
\ncom{\ds}{\displaystyle}
\ncom{\half}{\frac{1}{2}}
\ncom{\mbx}{\makebox{.25cm}}
\ncom{\hs}{\mbox{\hspace{.25cm}}}
\ncom{\rar}{\rightarrow}
\ncom{\Rar}{\Rightarrow}
\ncom{\noin}{\noindent}
\ncom{\bc}{\begin{center}}
\ncom{\ec}{\end{center}}
\ncom{\sz}{\scriptsize}
\ncom{\rf}{\ref}
\ncom{\s}{\sqrt{2}}
\ncom{\sgm}{\sigma}
\ncom{\Sgm}{\Sigma}
\ncom{\psgm}{\sigma^{\prime}}
\ncom{\dt}{\delta}
\ncom{\Dt}{\Delta}
\ncom{\lmd}{\lambda}
\ncom{\Lmd}{\Lambda}
\ncom{\Th}{\Theta}
\ncom{\e}{\eta}
\ncom{\eps}{\epsilon}
\ncom{\pcc}{\stackrel{P}{>}}
\ncom{\lp}{\stackrel{L_{p}}{>}}
\ncom{\dist}{{\rm\,dist}}
\ncom{\sspan}{{\rm\,span}}
\ncom{\re}{{\rm Re\,}}
\ncom{\im}{{\rm Im\,}}
\ncom{\sgn}{{\rm sgn\,}}
\ncom{\ba}{\begin{array}}
\ncom{\ea}{\end{array}}
\ncom{\hone}{\mbox{\hspace{1em}}}
\ncom{\htwo}{\mbox{\hspace{2em}}}
\ncom{\hthree}{\mbox{\hspace{3em}}}
\ncom{\hfour}{\mbox{\hspace{4em}}}
\ncom{\vone}{\vskip 2ex}
\ncom{\vtwo}{\vskip 4ex}
\ncom{\vonee}{\vskip 1.5ex}
\ncom{\vthree}{\vskip 6ex}
\ncom{\vfour}{\vspace*{8ex}}
\ncom{\norm}{\|\;\;\|}
\ncom{\integ}[4]{\int_{#1}^{#2}\,{#3}\,d{#4}}
\ncom{\vspan}[1]{{{\rm\,span}\{ #1 \}}}
\ncom{\dm}[1]{ {\displaystyle{#1} } }
\ncom{\ri}[1]{{#1} \index{#1}}

\newtheorem{theorem}{\bf Theorem}[section]
\newtheorem{remark}{\bf Remark}[section]

\newtheorem{corollary}{Corollary}[section]
\newtheorem{example}{Example}[section]

\newtheoremstyle
    {remarkstyle}
    {}
    {11pt}
    {}
    {}
    {\bfseries}
    {:}
    {     }
    {\thmname{#1} \thmnumber{#2} }

\theoremstyle{remarkstyle}

\def\eps{\varepsilon}

\begin{document}

\title{\Large E\lowercase{valuation of missing data mechanisms in two and three dimensional incomplete tables}}
\author[Sayan Ghosh]{S. Ghosh}
\address{Sayan Ghosh, Theoretical Statistics and Mathematics Unit,
 Indian Statistical Institute Kolkata, 203 B.T. Road, Kolkata 700108, INDIA.}
 \email{sayan38@gmail.com}
\author{P. Vellaisamy}
\address{P. Vellaisamy, Department of Mathematics,
Indian Institute of Technology Bombay, Powai, Mumbai 400076, INDIA.}
\email{pv@math.iitb.ac.in}

\begin{abstract}
The analysis of incomplete contingency tables is a practical and an interesting problem. In this paper, we provide characterizations for the various missing mechanisms of a variable in terms of response and non-response odds for two and three dimensional incomplete tables. Log-linear parametrization and some distinctive properties of the missing data models for the above tables are discussed. All possible cases in which data on one, two or all variables may be missing are considered. We study the missingness of each variable in a model, which is more insightful for analyzing cross-classified data than the missingness of the outcome vector. For sensitivity analysis of the incomplete tables, we propose easily verifiable procedures to evaluate the missing at random (MAR), missing completely at random (MCAR) and not missing at random (NMAR) assumptions of the missing data models. These methods depend only on joint and marginal odds computed from fully and partially observed counts in the tables, respectively. Finally, some real-life datasets are analyzed to illustrate our results, which are confirmed based on simulation studies.
\end{abstract}

\maketitle

\section{Introduction}
A contingency table with fully observed counts and supplemental margins (non-responses) is called an incomplete table. For inference purposes, three types of missing data mechanisms are used to study non-responses (see Little and Rubin (2002)): missing completely at random (MCAR), missing at random (MAR) and not missing at random (NMAR). The missing mechanism is said to be MCAR when missingness is independent of both observed and unobserved data, MAR when missingness depends only on observed data, and NMAR if missingness depends only on unobserved data. Non-responses can be either ignorable (when the missing data mechanism is MAR or MCAR, and the estimated parameters are distinct from those involving the missing data mechanism) or nonignorable (when the missing data mechanism is NMAR). 

The assumption regarding the missing data mechanism in the model cannot be usually confirmed from the model fit to the observed data. Hence, it is difficult to use non-response models to analyze incomplete tables. 
Molenberghs {\it et al.} (2008) considered the missingness of the outcome vector in an incomplete table and demonstrated that every NMAR model has a MAR counterpart with equal fit. Several researchers have implemented sensitivity analysis to assess the missing data mechanism in incomplete tables. One approach is to compare the relevant parameter estimates from a range of candidate models (see Baker {\it et al.} (1992)). Another approach is to consider overspecified models with sensitivity parameter and construct confidence intervals for the parameters to investigate the statistical uncertainty due to incomplete data and finite sampling (see Molenberghs {\it et al.} (2001) and Vansteelandt {\it et al.} (2006)). 

Park {\it et al.} (2014) identified sufficient conditions for the occurrence of boundary solutions (zero MLEs of cell probabilities) in two-way incomplete tables with both variables missing. Recently, Ghosh and Vellaisamy (2016a) studied boundary solutions in multidimensional incomplete tables with one or more variables missing, and established sufficient conditions for their occurrence. Also, Ghosh and Vellaisamy (2017) studied boundary solutions in two-way incomplete tables with both variables missing, and proved various results including necessary conditions for their occurrence. However, in this paper, we consider a different problem of evaluation of missing data mechanisms in two-way and three-way incomplete tables. So the methods suggested in this paper are also different, except the use of response and non-response odds and their estimators. Our goals here are to study characteristic properties of missing data models using the above odds and then develop sensitivity analysis based on the estimators to determine various missing mechanisms of the variables in some incomplete tables. Essentially, we provide conditions using only the observed data to assess the various missing data models without actually fitting them to the incomplete tables. Interestingly, the models suggested by our analysis are the ones with no boundary solutions, which is expected since boundary solutions pose a lot of problems for inference with missing data.      

Recently, Kim {\it et al.} (2015) proposed a new and convenient method of sensitivity analysis to assess the MAR assumption for a two-way incomplete table with one supplemental margin ($I\times J\times 2$ table in Baker and Laird (1988)). In this paper, we establish methods to assess the MCAR, MAR and NMAR assumptions for $I\times J\times 2\times 2$, $I\times J\times K\times 2$, $I\times J\times K\times 2\times 2$ and $I\times J\times K\times 2\times 2\times 2$ tables, that is, two-way tables with both variables missing, and three-way tables with one, two and all variables missing, respectively. Note that the evaluation of missing mechanisms in the above incomplete tables has not been studied earlier in the literature. The main advantage of our proposed methods is computational simplicity as they are based on the response and non-response odds involving only the observed counts or their sums in the tables. We also study distinguishing attributes of various missing data models and provide characterizations for the different missing mechanisms of a variable in terms of the above odds. These conditions help us to develop the assessment procedures. Another advantage of our methods is that they can suggest the missing mechanism of a variable (plausible models for the incomplete data) without computing the usual model selection criteria.    

The remaining part of the paper is organized as follows. Section 2 considers non-response log-linear models for $I\times J\times 2\times 2$ tables. Some results regarding characteristic features of these missing data models are provided. The missing mechanisms of the variables are identified using the response and non-response odds based on joint and marginal cell probabilities, respectively. Assessment of MAR, MCAR and NMAR mechanisms is carried out using the estimators of the above odds computed from only the observed counts in the tables. Therefore, there is no need to use numerical or simulation procedures for the analysis of such tables. The results and discussions in Section 2 are extended in Section 3 to three-way incomplete tables with one, two or all variables missing. Section 4 presents real-life data analysis examples along with bootstrapping to illustrate the results in Sections 2 and 3. Some concluding remarks about the methods of assessment for the missing mechanisms are provided in Section 5. Finally, the Appendix contains the proofs of the theorems in Sections 2 and 3.

\section{Missing data models for the $I\times J\times 2\times 2$ table}

Kim {\it et al.} (2015) considered missing data models and sensitivity analysis for the $I\times J\times 2$ table. They also mentioned that it would be of interest to study such models and develop sensitivity analysis for two-way tables with both variables subject to missingness. In this section, we address these issues.
 
Let $Y_{1}$ and $Y_{2}$ be two categorical variables with $I$ and $J$ levels respectively. It is assumed that data on both variables may be missing. For $i=1,2,$ let $R_{i}$ denote the missing indicator variable for $Y_{i}$ such that $R_{i}=1$ if $Y_{i}$ is observed and $R_{i}=2$ otherwise. Then we have an $I\times J\times 2\times 2$ table with cell probabilities ${\bf\pi} = \{\pi_{ijkl}\}$ and cell counts ${\bf y} = \{y_{ijkl}\},$ where $1\leq i\leq I,~1\leq j\leq J$ and $k,l = 1,2.$ The vector of observed frequencies is given by ${\bf y_{\textrm{obs}}} = \left(\{y_{ij11}\},\{y_{+j21}\},\{y_{i+12}\},y_{++22}\right),$ where a `+' in the subscript denotes summation over levels of the corresponding variable. Here $\{y_{ij11}\}$ are the fully observed counts, while $\{y_{+j21}\},\{y_{i+12}\}$ and $y_{++22}$ are the supplemental margins. Table \ref{t1} shows the $I\times J\times 2\times 2$ table. 
\begin{table}[ht]   
\caption{$I\times J\times 2\times 2$ Incomplete Table}\label{t1}
\begin{center}
$
\begin{array}{|c|c|cccc|c|}\hline
& & R_{2} = 1 & & & & R_{2} = 2 \\ \hline
& & Y_{2} = 1 & Y_{2} = 2 & \cdots & Y_{2} = J & Y_{2}~\textrm{missing} \\ \hline
R_{1} = 1 & Y_{1} = 1 & y_{1111} & y_{1211} & \cdots & y_{1J11} & y_{1+12} \\
& Y_{1} = 2 & y_{2111} & y_{2211} & \cdots & y_{2J11} & y_{2+12} \\ 
& \vdots & \vdots & \vdots & \vdots & \vdots & \vdots \\
& Y_{1} = I & y_{I111} & y_{I211} & \cdots & y_{IJ11} & y_{I+12} \\ \hline
R_{1} = 2 & Y_{1}~\textrm{missing} & y_{+121} & y_{+221} & \cdots & y_{+J21} & y_{++22} \\ \hline
\end{array}
$
\end{center}
\end{table}
Let the vector of expected counts be $\mu = \{\mu_{ijkl}\}$ and $N = \sum_{i,j,k,l}y_{ijkl}$ be the total cell count. Under Poisson sampling for observed cell counts, the log-likelihood of ${\bf\mu}$ is 
\begin{eqnarray}\label{eq2.1}
l({\bf\mu};{\bf y}_{\textrm{obs}}) &=& \sum_{i,j}y_{ij11}\log \mu_{ij11} + \sum_{j}y_{+j21}\log \mu_{+j21} + \sum_{i}y_{i+12}\log \mu_{i+12} \nonumber \\ 
& & + y_{++22}\log \mu_{++22} - \sum_{i,j,k,l}\mu_{ijkl} + \Delta,
\end{eqnarray} 
where $\Delta$ is some constant. For an $I\times J\times 2\times 2$ incomplete table, Baker {\it et al.} (1992) proposed the following log-linear model (with no three-way or four-way interactions):
\begin{eqnarray}\label{eq1}
\log \mu_{ijkl} &=& \lambda + \lambda_{Y_{1}}(i) + \lambda_{Y_{2}}(j) + \lambda_{R_{1}}(k) + \lambda_{R_{2}}(l) + \lambda_{Y_{1}Y_{2}}(i,j) \nonumber \\
& & + \lambda_{Y_{1}R_{1}}(i,k) + \lambda_{Y_{2}R_{1}}(j,k) + \lambda_{Y_{1}R_{2}}(i,l) + \lambda_{Y_{2}R_{2}}(j,l) + \lambda_{R_{1}R_{2}}(k,l). 
\end{eqnarray}
Each log-linear parameter in (\ref{eq1}) satisfies the constraint that the sum over each of its arguments ($i,j,k,l$) is 0. Henceforth, in this paper, we study the missingness of each variable and not the usual missingness of the outcome vector as considered by Baker {\it et al.} (1992). Our approach is based on the classification scheme of missing data models considered by Park {\it et al.} (2014). Also, this approach is more natural and meaningful than the conventional one in the context of incomplete tables since it provides an insight into each variable's missing mechanism, and hence allows us to consider a larger class of models with explicit forms.

By definition, the missing mechanism of a variable may depend on itself (NMAR) or on another observed variable (MAR) or none (MCAR). Equivalently, for $i\neq j$ in (\ref{eq1}), we have $\lambda_{Y_{j}R_{i}}=0$, $\lambda_{Y_{i}R_{i}}=0$, and $\lambda_{Y_{i}R_{i}}=\lambda_{Y_{j}R_{i}}=0$ if the missing mechanism of $Y_{i}$ is NMAR, MAR and MCAR respectively. By interchanging $i$ and $j$, similar constraints in (\ref{eq1}) can be imposed for the various missing mechanisms of $Y_{j}$. The missing data models can thus be obtained as submodels of (\ref{eq1}). Baker {\it et al.} (1992) suggested nine such identifiable models (using a different parametrization), whose log-linear formulations (based on different missing mechanisms for $Y_{1}$ and $Y_{2}$) are as follows (see Park {\it et al.} (2014) or Ghosh and Vellaisamy (2017)). 
\vone\noindent
M1. NMAR for $Y_{1}$, MCAR for $Y_{2}$ : 
\begin{equation*}
\log \mu_{ijkl} = \lambda + \lambda_{Y_{1}}(i) + \lambda_{Y_{2}}(j) + \lambda_{R_{1}}(k) + \lambda_{R_{2}}(l) + \lambda_{Y_{1}Y_{2}}(i,j) + \lambda_{Y_{1}R_{1}}(i,k) + \lambda_{R_{1}R_{2}}(k,l)
\end{equation*}
M2. NMAR for $Y_{1}$, MAR for $Y_{2}$ : 
\begin{equation*}
\log \mu_{ijkl} = \lambda + \lambda_{Y_{1}}(i) + \lambda_{Y_{2}}(j) + \lambda_{R_{1}}(k) + \lambda_{R_{2}}(l) + \lambda_{Y_{1}Y_{2}}(i,j) + \lambda_{Y_{1}R_{1}}(i,k) + \lambda_{Y_{1}R_{2}}(i,l) + \lambda_{R_{1}R_{2}}(k,l)
\end{equation*}
M3. NMAR for both $Y_{1}$ and $Y_{2}$ : 
\begin{equation*}
\log \mu_{ijkl} = \lambda + \lambda_{Y_{1}}(i) + \lambda_{Y_{2}}(j) + \lambda_{R_{1}}(k) + \lambda_{R_{2}}(l) + \lambda_{Y_{1}Y_{2}}(i,j) + \lambda_{Y_{1}R_{1}}(i,k) + \lambda_{Y_{2}R_{2}}(j,l) + \lambda_{R_{1}R_{2}}(k,l)
\end{equation*}
M4. MAR for $Y_{1}$, MCAR for $Y_{2}$ : 
\begin{equation*}
\log \mu_{ijkl} = \lambda + \lambda_{Y_{1}}(i) + \lambda_{Y_{2}}(j) + \lambda_{R_{1}}(k) + \lambda_{R_{2}}(l) + \lambda_{Y_{1}Y_{2}}(i,j) + \lambda_{Y_{2}R_{1}}(j,k) + \lambda_{R_{1}R_{2}}(k,l)
\end{equation*}
M5. MAR for both $Y_{1}$ and $Y_{2}$ : 
\begin{equation*}
\log \mu_{ijkl} = \lambda + \lambda_{Y_{1}}(i) + \lambda_{Y_{2}}(j) + \lambda_{R_{1}}(k) + \lambda_{R_{2}}(l) + \lambda_{Y_{1}Y_{2}}(i,j) + \lambda_{Y_{2}R_{1}}(j,k) + \lambda_{Y_{1}R_{2}}(i,l) + \lambda_{R_{1}R_{2}}(k,l)
\end{equation*}
M6. MAR for $Y_{1}$, NMAR for $Y_{2}$ : 
\begin{equation*}
\log \mu_{ijkl} = \lambda + \lambda_{Y_{1}}(i) + \lambda_{Y_{2}}(j) + \lambda_{R_{1}}(k) + \lambda_{R_{2}}(l) + \lambda_{Y_{1}Y_{2}}(i,j) + \lambda_{Y_{2}R_{1}}(j,k) + \lambda_{Y_{2}R_{2}}(j,l) + \lambda_{R_{1}R_{2}}(k,l)
\end{equation*}
M7. MCAR for $Y_{1}$, MAR for $Y_{2}$ : 
\begin{equation*}
\log \mu_{ijkl} = \lambda + \lambda_{Y_{1}}(i) + \lambda_{Y_{2}}(j) + \lambda_{R_{1}}(k) + \lambda_{R_{2}}(l) + \lambda_{Y_{1}Y_{2}}(i,j) + \lambda_{Y_{1}R_{2}}(i,l) + \lambda_{R_{1}R_{2}}(k,l)
\end{equation*}
M8. MCAR for $Y_{1}$, NMAR for $Y_{2}$ :
\begin{equation*}
\log \mu_{ijkl} = \lambda + \lambda_{Y_{1}}(i) + \lambda_{Y_{2}}(j) + \lambda_{R_{1}}(k) + \lambda_{R_{2}}(l) + \lambda_{Y_{1}Y_{2}}(i,j) + \lambda_{Y_{2}R_{2}}(j,l) + \lambda_{R_{1}R_{2}}(k,l)
\end{equation*}
M9. MCAR for both $Y_{1}$ and $Y_{2}$ : 
\begin{equation*}
\log \mu_{ijkl} = \lambda + \lambda_{Y_{1}}(i) + \lambda_{Y_{2}}(j) + \lambda_{R_{1}}(k) + \lambda_{R_{2}}(l) + \lambda_{Y_{1}Y_{2}}(i,j) + \lambda_{R_{1}R_{2}}(k,l)
\end{equation*}
Next, we describe various features of these models, which help us investigate the missing data mechanisms of the variables in an $I\times J\times 2\times 2$ table.

\subsection{Properties of the missing data models}

Define the following odds for a pair $(j,j')$ of $Y_{2}$ :
\begin{equation*}
\nu_{i}(j,j') = \frac{\pi_{ij11}}{\pi_{ij'11}},~\nu_{n}(j,j') = \min_{i}\{\nu_{i}(j,j')\},~\nu_{m}(j,j') = \max_{i}\{\nu_{i}(j,j')\},~\nu(j,j') = \frac{\pi_{+j21}}{\pi_{+j'21}},
\end{equation*}
where $\nu_{i}(j,j')$ are the response odds and $\nu(j,j')$ are the non-response odds when $Y_{1}$ is missing. Also, define the following odds for any pair $(i,i')$ of $Y_{1}$ :
\begin{equation*}
\omega_{j}(i,i') = \frac{\pi_{ij11}}{\pi_{i'j11}},~\omega_{n}(i,i') = \min_{j}\{\omega_{j}(i,i')\},~\omega_{m}(i,i') = \max_{j}\{\omega_{j}(i,i')\},~\omega(i,i') = \frac{\pi_{i+12}}{\pi_{i'+12}}.
\end{equation*}
Here, $\omega_{j}(i,i')$ are the response odds, while $\omega(i,i')$ are the non-response odds when $Y_{2}$ is missing. Let $OI(i,i') = (\omega_{n}(i,i'),~\omega_{m}(i,i'))$ and $OI(j,j') = (\nu_{n}(j,j'),~\nu_{m}(j,j'))$. Note that the interval $OI(i,i')$ contains $\omega_{j}(i,i')$, while the interval $OI(j,j')$ contains $\nu_{i}(j,j')$. 

We now study the behaviour of the above odds under Models M1-M9. More specifically, we investigate the conditions under which the non-response odds belong to the open intervals formed by the response odds. These conditions help us to characterize the MAR, MCAR and NMAR mechanisms of a variable, which are useful for their assessment and hence model selection based only on the observed cell counts in the incomplete table. 

Let $m\in\{1,\ldots,J\}$ denote the level of $Y_{2}$ corresponding to $\omega_{m}(i,i')$. Then define 
\begin{equation}\label{b1} 
B_{m}(i,i') = \frac{\sum_{j}\exp\{\lambda_{Y_{2}}(j) + \lambda_{Y_{1}Y_{2}}(i',j) + \lambda_{Y_{1}Y_{2}}(i,m) + \lambda_{Y_{2}R_{1}}(j,1)\}}{\sum_{j}\exp\{\lambda_{Y_{2}}(j) + \lambda_{Y_{1}Y_{2}}(i,j) + \lambda_{Y_{1}Y_{2}}(i',m) + \lambda_{Y_{2}R_{1}}(j,1)\}}.  
\end{equation}
Let $B^{\ast}_{m}(i,i')$ denote $B_{m}(i,i')$ under Model M5 and $B_{m}(i,i')$ with $\lambda_{Y_{2}R_{1}}(j,1)=0$ under Models M2 and M7. Similarly, define $B^{\ast}_{n}(i,i')$. From Appendix A, we observe that $B^{\ast}_{m}(i,i') > 1$ and $B^{\ast}_{n}(i,i') < 1$ for any pair $(i,i')$ of $Y_{1}$. Then the next theorem characterizes the missing data mechanisms of $Y_{2}$ in an $I\times J\times 2\times 2$ table.
\begin{theorem}\label{th2.2}
Under Models M1-M9 for an $I\times J\times 2\times 2$ table, we have the following cases corresponding to the missing mechanism of $Y_{2}$. 
\begin{enumerate}
\item[(a)] If $Y_{2}$ has a MCAR or NMAR mechanism, then $\omega(i,i')\in OI(i,i')$ if 
$|\lambda_{Y_{2}R_{2}}(j,2)| < \infty$. 
\item[(b)] If $Y_{2}$ has a MAR mechanism, then only one of the following conditions holds for each pair $(i,i')$ of $Y_{1}$ : \\ 
(i) $\omega(i,i')\in OI(i,i')$ iff $-\frac{1}{2}\log B^{\ast}_{m}(i,i') < \lambda_{Y_{1}R_{2}}(i',2) - \lambda_{Y_{1}R_{2}}(i,2) < -\frac{1}{2}\log B^{\ast}_{n}(i,i')$, \\
(ii) $\omega(i,i')\not\in OI(i,i')$ iff $\lambda_{Y_{1}R_{2}}(i',2) - \lambda_{Y_{1}R_{2}}(i,2) > -\frac{1}{2}\log B^{\ast}_{n}(i,i')$ or $\lambda_{Y_{1}R_{2}}(i',2) - \lambda_{Y_{1}R_{2}}(i,2)  < -\frac{1}{2}\log B^{\ast}_{m}(i,i')$.
\end{enumerate}
\end{theorem} 
\begin{proof}
See Appendix A.
\end{proof}
Now let $m\in\{1,\ldots,I\}$ denote the level of $Y_{1}$ corresponding to $\nu_{m}(j,j')$. Then define 
\begin{equation}\label{a1}
A_{m}(j,j') = \frac{\sum_{i}\exp\{\lambda_{Y_{1}}(i) + \lambda_{Y_{1}Y_{2}}(i,j') + \lambda_{Y_{1}Y_{2}}(m,j) + \lambda_{Y_{1}R_{2}}(i,1)\}}{\sum_{i}\exp\{\lambda_{Y_{1}}(i) + \lambda_{Y_{1}Y_{2}}(i,j) + \lambda_{Y_{1}Y_{2}}(m,j') + \lambda_{Y_{1}R_{2}}(i,1)\}}.
\end{equation}
Let $A^{\ast}_{m}(j,j')$ denote $A_{m}(j,j')$ under Model M5 and $A_{m}(j,j')$ with $\lambda_{Y_{1}R_{2}}(i,1)=0$ under Models M4 and M6. Similarly, define $A^{\ast}_{n}(j,j')$. From Appendix B, we have $A^{\ast}_{m}(j,j') > 1$ and $A^{\ast}_{n}(j,j') < 1$ for any pair $(j,j')$ of $Y_{2}$. Then we have the following characterization for the missing data mechanisms of $Y_{1}$ in an $I\times J\times 2\times 2$ table.

\begin{theorem}\label{th2.1}
Under Models M1-M9 for an $I\times J\times 2\times 2$ table, we have the following cases corresponding to the missing mechanism of $Y_{1}$.
\begin{enumerate}
\item[(a)] If $Y_{1}$ has a MCAR or NMAR mechanism, then $\nu(j,j')\in OI(j,j')$ if $|\lambda_{Y_{1}R_{1}}(i,2)| < \infty$. 
\item[(b)] If $Y_{1}$ has a MAR mechanism, then only one of the following conditions holds for each pair $(j,j')$ of $Y_{2}$ : \\
(i) $\nu(j,j')\in OI(j,j')$ iff $-\frac{1}{2}\log A^{\ast}_{m}(j,j') < \lambda_{Y_{2}R_{1}}(j',2) - \lambda_{Y_{2}R_{1}}(j,2) < -\frac{1}{2}\log A^{\ast}_{n}(j,j')$, \\
(ii) $\nu(j,j')\not\in OI(j,j')$ iff $\lambda_{Y_{2}R_{1}}(j',2) - \lambda_{Y_{2}R_{1}}(j,2) > -\frac{1}{2}\log A^{\ast}_{n}(j,j')$ or $\lambda_{Y_{2}R_{1}}(j',2) - \lambda_{Y_{2}R_{1}}(j,2)  < -\frac{1}{2}\log A^{\ast}_{m}(j,j')$.
\end{enumerate} 
\end{theorem}
\begin{proof}
See Appendix B.
\end{proof}

\begin{remark}\label{rem13}
From Theorems \ref{th2.2}(a) and \ref{th2.1}(a), note that if the missing mechanism of $Y_{1}$ or $Y_{2}$ is NMAR or MCAR, then $\nu(j,j')\in OI(j,j')$ for any pair $(j,j')$ of $Y_{2}$ and $\omega(i,i')\in OI(i,i')$ for any pair $(i,i')$ of $Y_{1}$. Also, from Theorem \ref{th2.2}(b), if there exists at least one pair $(i,i')$ of $Y_{1}$ such that $\omega(i,i')\not\in OI(i,i')$, then $|\lambda_{Y_{1}R_{2}}(i',2) - \lambda_{Y_{1}R_{2}}(i,2)|$ is larger than that when $\omega(i,i')\in OI(i,i')$. We say that the missing mechanism of $Y_{2}$ is strong MAR in the first case and non-strong (weak) in the second one. Similar results hold for the MAR mechanism of $Y_{1}$. 
\end{remark}
\begin{remark}\label{rem14}
If $I = J = 2$, then we have a $2\times 2\times 2\times 2$ table. Under Models M1-M9,
\begin{equation*}
\lambda_{Y_{1}Y_{2}}(1,1) = \frac{1}{4}\log\left[\frac{\nu_{1}(1,2)}{\nu_{2}(1,2)}\right] = \frac{1}{4}\log\left[\frac{\omega_{1}(1,2)}{\omega_{2}(1,2)}\right].
\end{equation*}
Hence, for fixed $\pi$, the length of $OI(1,2) = |\nu_{1}(1,2)-\nu_{2}(1,2)|$ and that of $OI'(1,2) = |\omega_{1}(1,2)-\omega_{2}(1,2)|$ or equivalently the sizes of the parameter regions for the weak MAR mechanisms of $Y_{1}$ and $Y_{2}$ respectively are each directly proportional to $|\lambda_{Y_{1}Y_{2}}(1,1)|$, the strength of the association between $Y_{1}$ and $Y_{2}$.     
\end{remark}

\subsection{Assessment of the MCAR, NMAR and MAR mechanisms}

A perfect fit model is one in which the estimated expected counts are equal to the observed counts. It is known that Models M1, M4, M7, M8 and M9 do not provide perfect fits for observed counts in the tables (see Table II on p.~647 of Baker {\it et al.} (1992)). However, Models M2, M3, M5 and M6 are perfect fit models so that $\widehat{\pi}_{ij11} = y_{ij11}/N$, $\widehat{\pi}_{i+12} = y_{i+12}/N$ and $\widehat{\pi}_{+j21} = y_{+j21}/N$ (see Table II on p.~648 of Baker {\it et al.} (1992)). Hence, the estimators of the various odds under them are as follows. 
\begin{eqnarray*}
\widehat{\nu}_{i}(j,j') &=& \frac{y_{ij11}}{y_{ij'11}},~\widehat{\nu}_{n}(j,j') = \min_{i}\{\widehat{\nu}_{i}(j,j')\},~\widehat{\nu}_{m}(j,j') = \max_{i}\{\widehat{\nu}_{i}(j,j')\},~\widehat{\nu}(j,j') = \frac{y_{+j21}}{y_{+j'21}}; \nonumber \\
\widehat{\omega}_{j}(i,i') &=& \frac{y_{ij11}}{y_{i'j11}},~\widehat{\omega}_{n}(i,i') = \min_{j}\{\widehat{\omega}_{j}(i,i')\},~\widehat{\omega}_{m}(i,i') = \max_{j}\{\widehat{\omega}_{j}(i,i')\},~\widehat{\omega}(i,i') = \frac{y_{i+12}}{y_{i'+12}}. \nonumber
\end{eqnarray*}
Note that the estimated expected counts and hence the MLE's of the response and the non-response odds under non-perfect fit models are non-trivial functions of the observed counts in the tables. For example, the estimated cell probabilities under Model M1 (see p.~647 in Baker {\it et al.} (1992)) are
\begin{equation*}
\hat{\pi}_{ij11} = \frac{y_{ij11}y_{i+1+}y_{++11}}{Ny_{i+11}y_{++1+}},\quad\hat{\pi}_{+j21}=\frac{y_{+j21}}{N},\quad\hat{\pi}_{i+12} = \frac{y_{++12}\sum_{j}\hat{\pi}_{ij11}}{y_{++11}} = \frac{y_{i+1+}y_{++12}}{Ny_{++1+}}.
\end{equation*}
Hence, the MLE's of the response and non-response odds under Model M1 are
\begin{eqnarray*}
\widehat{\nu}_{i}(j,j') &=& \frac{y_{ij11}}{y_{ij'11}},~\widehat{\nu}_{n}(j,j') = \min_{i}\{\widehat{\nu}_{i}(j,j')\},~\widehat{\nu}_{m}(j,j') = \max_{i}\{\widehat{\nu}_{i}(j,j')\},~\widehat{\nu}(j,j') = \frac{y_{+j21}}{y_{+j'21}}; \nonumber \\
\widehat{\omega}_{j}(i,i') &=& \frac{y_{ij11}y_{i+1+}y_{i'+11}}{y_{i'j11}y_{i+11}y_{i'+1+}},~\widehat{\omega}_{n}(i,i') = \min_{j}\{\widehat{\omega}_{j}(i,i')\},~\widehat{\omega}_{m}(i,i') = \max_{j}\{\widehat{\omega}_{j}(i,i')\}, \nonumber \\
\widehat{\omega}(i,i') &=& \frac{\sum_{j}{\hat{\pi}_{ij11}}}{\sum_{j}\hat{\pi}_{i'j11}} = \frac{y_{i+1+}}{y_{i'+1+}}. \nonumber
\end{eqnarray*}
Using Table II on p.~647 of Baker {\it et al.} (1992), the MLE's of the response and non-response odds under other non-perfect fit models  may be obtained. It can be shown that the estimators of $\nu$'s for Models M1 and M4 are the same as those for the perfect fit models, while the estimators of $\omega$'s for Models M7 and M8 and those for the the perfect fit models are identical. Now let $\widehat{OI}(i,i') = (\widehat{\omega}_{n}(i,i'),~\widehat{\omega}_{m}(i,i'))$ and $\widehat{OI}(j,j') = (\widehat{\nu}_{n}(j,j'),~\widehat{\nu}_{m}(j,j'))$. Then the corollary below follows from Theorems \ref{th2.1}(a) and \ref{th2.2}(a), and Remark \ref{rem13}.
\begin{corollary}\label{cor2.1}
For an $I\times J\times 2\times 2$ table, if $\widehat{\omega}(i,i')\not\in\widehat{OI}(i,i')$ for at least one pair $(i,i')$ of $Y_{1}$ or $\widehat{\nu}(j,j')\not\in\widehat{OI}(j,j')$ for at least one pair $(j,j')$ of $Y_{2}$, then the plausible missing data mechanism of $Y_{1}$ or $Y_{2}$ is MAR and not NMAR or MCAR. 
\end{corollary}

\begin{remark}\label{rem15} 
Note that only observed counts or their functions are used in Corollary \ref{cor2.1} to assess the MCAR, MAR and NMAR mechanisms of the variables in an $I\times J\times 2\times 2$ table. Also, our results in this section along with those obtained by Kim {\it et al.} (2015) completely characterize the missing mechanisms of variables in two-way incomplete tables.
\end{remark}

\section{Missing data models for three-way incomplete tables}

In this section, we propose log-linear models for three-way incomplete tables and study missing data mechanisms of the variables using these models. All possible cases in which data on one, two or all variables may be missing are considered. We also develop sensitivity analysis for such tables. Suppose $Y_{1}$, $Y_{2}$ and $Y_{3}$ are three categorical variables with $I$, $J$ and $K$ levels respectively. Then we have the following cases.

\subsection{Case 1: Missing in one of the variables}
 
Without loss of generality (WLOG), let data on $Y_{1}$ be missing and $R$ denote the missing indicator for $Y_{1}$ such that $R = 1$ if $Y_{1}$ is observed and $R = 2$ otherwise. Then for $Y_{1},~Y_{2},~Y_{3}$ and $R,$ we have an $I\times J\times K\times 2$ table with cell counts ${\bf y} = \{y_{ijkr}\}$, where $1\leq i\leq I,~1\leq j\leq J,~1\leq k\leq K$ and $r = 1,2.$ The vector of observed counts is ${\bf y_{\textrm{obs}}} = (\{y_{ijk1}\},\{y_{+jk2}\})$, where $\{y_{ijk1}\}$ are the fully observed counts and $\{y_{+jk2}\}$ are the supplementary margins with `+' representing summation over levels of the corresponding variable. 
For $I=J=K=2$, the $2\times 2\times 2\times 2$ incomplete table is given by Table \ref{t4}. 
\begin{table}[ht]
\caption{$2\times 2\times 2\times 2$ Incomplete Table}\label{t4}
\begin{center}
$
\begin{array}{|c|c|c|cc|}\hline
& & & Y_{3} = 1 & Y_{3} = 2 \\ \hline
R = 1 & Y_{1} = 1 & Y_{2} = 1 & y_{1111} & y_{1121} \\   
& & Y_{2} = 2 & y_{1211} & y_{1221} \\ \hline 
& Y_{1} = 2 & Y_{2} = 1 & y_{2111} & y_{2121} \\   
& & Y_{2} = 2 & y_{2211} & y_{2221} \\ \hline 
R = 2 & \text{Missing} & Y_{2} = 1 & y_{+112} & y_{+122} \\   
& & Y_{2} = 2 & y_{+212} & y_{+222} \\ \hline 
\end{array}
$
\end{center} 
\end{table}
Let ${\bf\pi} = \{\pi_{ijkr}\}$ be the vector of cell probabilities, $\mu = \{\mu_{ijkr}\}$ be the vector of expected counts and $N = \sum_{i,j,k,r} y_{ijkr}$ be the total cell count. Under Poisson sampling for observed cell counts, the log-likelihood kernel of ${\bf\mu}$ is 
\begin{equation}\label{eq3.1}
l({\bf\mu};{\bf y}_{\textrm{obs}}) = \sum_{i,j,k}y_{ijk1}\log \mu_{ijk1} + \sum_{j,k}y_{+jk2}\log \mu_{+jk2} - \sum_{i,j,k,r}\mu_{ijkr}.
\end{equation}
The log-linear model is (see Ghosh and Vellaisamy (2016a)) 
\begin{eqnarray}\label{eq2}
\log \mu_{ijkr} &=& \lambda + \lambda_{Y_{1}}(i) + \lambda_{Y_{2}}(j) + \lambda_{Y_{3}}(k) + \lambda_{R}(r) +\lambda_{Y_{1}Y_{2}}(i,j) + \lambda_{Y_{1}Y_{3}}(i,k) + \lambda_{Y_{2}Y_{3}}(j,k) \nonumber \\
& &  + \lambda_{Y_{1}R}(i,r) + \lambda_{Y_{2}R}(j,r) + \lambda_{Y_{3}R}(k,r). 
\end{eqnarray}
We avoid higher order interactions in (\ref{eq2}) since they are difficult to interpret and obtaining closed-form maximum likelihood estimates of the parameters becomes complicated. Each log-linear parameter in (\ref{eq2}) satisfies the constraint that the sum over each of its arguments is 0. It is assumed in this case and subsequent ones that the missing mechanism of a variable may depend on itself (NMAR) or on one of the other variables (MAR) or none (MCAR). Accordingly, the various missing data models, which are submodels of (\ref{eq2}), are as follows. 
\vone\noindent
C1. NMAR for $Y_{1}$ :
\begin{equation*}
\log\mu_{ijkr} = \lambda + \lambda_{Y_{1}}(i) + \lambda_{Y_{2}}(j) + \lambda_{Y_{3}}(k) + \lambda_{R}(r) + \lambda_{Y_{1}Y_{2}}(i,j) + \lambda_{Y_{1}Y_{3}}(i,k) + \lambda_{Y_{2}Y_{3}}(j,k) + \lambda_{Y_{1}R}(i,r)
\end{equation*}
C2. MAR for $Y_{1}$ (missing mechanism depends on $Y_{2}$) :
\begin{equation*}
\log\mu_{ijkr} = \lambda + \lambda_{Y_{1}}(i) + \lambda_{Y_{2}}(j) + \lambda_{Y_{3}}(k) + \lambda_{R}(r) + \lambda_{Y_{1}Y_{2}}(i,j) + \lambda_{Y_{1}Y_{3}}(i,k) + \lambda_{Y_{2}Y_{3}}(j,k) + \lambda_{Y_{2}R}(j,r)
\end{equation*}
C3. MAR for $Y_{1}$ (missing mechanism depends on $Y_{3}$) :
\begin{equation*}
\log\mu_{ijkr} = \lambda + \lambda_{Y_{1}}(i) + \lambda_{Y_{2}}(j) + \lambda_{Y_{3}}(k) + \lambda_{R}(r) + \lambda_{Y_{1}Y_{2}}(i,j) + \lambda_{Y_{1}Y_{3}}(i,k) + \lambda_{Y_{2}Y_{3}}(j,k) + \lambda_{Y_{3}R}(k,r)
\end{equation*}
C4. MCAR for $Y_{1}$ :
\begin{equation*}
\log\mu_{ijkr} = \lambda + \lambda_{Y_{1}}(i) + \lambda_{Y_{2}}(j) + \lambda_{Y_{3}}(k) + \lambda_{R}(r) + \lambda_{Y_{1}Y_{2}}(i,j) + \lambda_{Y_{1}Y_{3}}(i,k) + \lambda_{Y_{2}Y_{3}}(j,k)
\end{equation*} 
Note that for each of the above models, there is an association term between a variable and its missing indicator if the missing mechanism is NMAR for that variable (for example, the term $\lambda_{Y_{1}R}(i,r)$ in Model C1), between its missing indicator and some other variable if the missing mechanism is MAR for that variable (for example, the term $\lambda_{Y_{2}R}(j,r)$ in Model C2) and none if the missing mechanism is MCAR for a variable (for example, $\lambda_{Y_{1}R}(i,r)$, $\lambda_{Y_{2}R}(j,r)$ and $\lambda_{Y_{3}R}(k,r)$ are absent in Model C4). This follows from the definitions of the missing mechanisms.

\subsubsection{Properties of the missing data models}
 
Define the following odds for any pair $(j,j')$ of $Y_{2}$ and $1\leq k\leq K$ :
\begin{equation*}
\nu_{ik}(j,j') = \frac{\pi_{ijk1}}{\pi_{ij'k1}},~\nu_{nk}(j,j') = \min_{i}\{\nu_{ik}(j,j')\},~\nu_{mk}(j,j') = \max_{i}\{\nu_{ik}(j,j')\},~\nu_{k}(j,j') = \frac{\pi_{+jk2}}{\pi_{+j'k2}}.
\end{equation*}
Similarly, define the following odds for any pair $(k,k')$ of $Y_{3}$ and $1\leq j\leq J$ :
\begin{equation*}
\nu_{ij}(k,k') = \frac{\pi_{ijk1}}{\pi_{ijk'1}},~\nu_{nj}(k,k') = \min_{i}\{\nu_{ij}(k,k')\},~\nu_{mj}(k,k') = \max_{i}\{\nu_{ij}(k,k')\},~\nu_{j}(k,k') = \frac{\pi_{+jk2}}{\pi_{+jk'2}}.
\end{equation*}
Let $OI_{k}(j,j') = (\nu_{nk}(j,j'),~\nu_{mk}(j,j'))$ and $OI_{j}(k,k') = (\nu_{nj}(k,k'),~\nu_{mj}(k,k'))$. Then the following two results characterize the various missing mechanisms of a variable in an $I\times J\times K\times 2$ table, which prove useful for their evaluation and hence model selection based only on the observed cell counts. 
\begin{theorem}\label{th3.1}
Suppose $Y_{1}$ has a NMAR or MCAR mechanism in an $I\times J\times K\times 2$ table. Then $\nu_{k}(j,j')\in OI_{k}(j,j')$ and $\nu_{j}(k,k')\in OI_{j}(k,k')$ if $|\lambda_{Y_{1}R}(i,2)| < \infty$. 
\end{theorem}
\begin{proof}
See Appendix C.
\end{proof}
Let $m\in\{1,\ldots,I\}$ be the level of $Y_{1}$ corresponding to $\nu_{mk}(j,j')$. Then define 
\begin{equation}\label{a3}
A_{mk}(j,j') = \frac{\sum_{i}\exp\{\lambda_{Y_{1}}(i) + \lambda_{Y_{1}Y_{2}}(i,j') + \lambda_{Y_{1}Y_{2}}(m,j) + \lambda_{Y_{1}Y_{3}}(i,k)\}}{\sum_{i}\exp\{\lambda_{Y_{1}}(i) + \lambda_{Y_{1}Y_{2}}(i,j) + \lambda_{Y_{1}Y_{2}}(m,j') + \lambda_{Y_{1}Y_{3}}(i,k)\}}.
\end{equation}
Also, let $A_{mj}(k,k')$ be obtained from (\ref{a3}) by interchanging $j$ with $k$ and $j'$ with $k'$. Similarly, define $A_{nk}(j,j')$ and $A_{nj}(k,k')$. From Appendix D, we observe that $A_{mk}(j,j') > 1$ and $A_{mj}(k,k') > 1$, while $A_{nk}(j,j') < 1$ and $A_{nj}(k,k') < 1$. Henceforth, Condition (L1, L2) holds means both conditions L1 and L2 hold. 
\begin{theorem}\label{th3.2}
Suppose $Y_{1}$ has a MAR mechanism in an $I\times J\times K\times 2$ table. Then only one of the Conditions (1a,1b) and (2a,2b)  holds: 
\begin{enumerate}
\item[1a.] For each $k$ and each pair $(j,j')$ of $Y_{2}$, only one of the conditions below holds: \\
(i) $\nu_{k}(j,j')\in OI_{k}(j,j')$ iff $-\frac{1}{2}\log A_{mk}(j,j') < \lambda_{Y_{2}R}(j',2) - \lambda_{Y_{2}R}(j,2) < -\frac{1}{2}\log A_{nk}(j,j')$, \\
(ii) $\nu_{k}(j,j')\not\in OI_{k}(j,j')$ iff $\lambda_{Y_{2}R}(j',2) - \lambda_{Y_{2}R}(j,2) > -\frac{1}{2}\log A_{nk}(j,j')$ or $\lambda_{Y_{2}R}(j',2) - \lambda_{Y_{2}R}(j,2)  < -\frac{1}{2}\log A_{mk}(j,j')$,
\item[1b.] $\nu_{j}(k,k')\in OI_{j}(k,k')$; 
\item[2a.] $\nu_{k}(j,j')\in OI_{k}(j,j')$, 
\item[2b.] For each $j$ and each pair $(k,k')$ of $Y_{3}$, only one of the conditions below holds: \\
(i) $\nu_{j}(k,k')\in OI_{j}(k,k')$ iff $-\frac{1}{2}\log A_{mj}(k,k') < \lambda_{Y_{3}R}(k',2) - \lambda_{Y_{3}R}(k,2) < -\frac{1}{2}\log A_{nj}(k,k'),$ \\
(ii) $\nu_{j}(k,k')\not\in OI_{j}(k,k')$ iff $\lambda_{Y_{3}R}(k',2) - \lambda_{Y_{3}R}(k,2) > -\frac{1}{2}\log A_{nj}(k,k')$ or $\lambda_{Y_{3}R}(k',2) - \lambda_{Y_{3}R}(k,2)  < -\frac{1}{2}\log A_{mj}(k,k')$. 
\end{enumerate}
\end{theorem}
\begin{proof}
See Appendix D.
\end{proof}

\begin{remark}\label{rem23}
From Theorem \ref{th3.1}, note that if the missing mechanism of $Y_{1}$ is NMAR or MCAR, then $\nu_{k}(j,j')\in OI_{k}(j,j')$ for any pair $(j,j')$ of $Y_{2}$ and $\nu_{j}(k,k')\in OI_{j}(k,k')$ for any pair $(k,k')$ of $Y_{3}$. Also, from Theorem \ref{th3.2} (1a), if there exists at least one $k$ or at least one pair $(j,j')$ of $Y_{2}$ such that $\nu_{k}(j,j')\not\in OI_{k}(j,j')$, then $|\lambda_{Y_{2}R}(j',2) - \lambda_{Y_{2}R}(j,2)|$ is larger than that when $\nu_{k}(j,j')\in OI_{k}(j,j')$. We say that the missing mechanism of $Y_{1}$ is strong MAR in the first case and non-strong (weak) in the second one.
Similar results follow when we consider $\lambda_{Y_{3}R}(k,2)$'s in Theorem \ref{th3.2} (2b).
\end{remark}
\begin{remark}\label{rem24}
If $I = J = K = 2$, then we have a $2\times 2\times 2\times 2$ table. Under Models C2 and C3, we have
\begin{equation*}
\lambda_{Y_{1}Y_{2}}(1,1) = \frac{1}{4}\log\left[\frac{\nu_{1k}(1,2)}{\nu_{2k}(1,2)}\right],~\lambda_{Y_{1}Y_{3}}(1,1) = \frac{1}{4}\log\left[\frac{\nu_{1j}(1,2)}{\nu_{2j}(1,2)}\right].
\end{equation*}
Hence, for fixed $\pi$, the length of $OI_{k}(1,2) = |\nu_{1k}(1,2)-\nu_{2k}(1,2)|$ or that of $OI_{j}(1,2) = |\nu_{1j}(1,2)-\nu_{2j}(1,2)|$, that is, the size of the parameter region for the weak MAR mechanism of $Y_{1}$ is directly proportional to $|\lambda_{Y_{1}Y_{2}}(1,1)|$ or $|\lambda_{Y_{1}Y_{3}}(1,1)|$, the strength of the associations between $Y_{1}$ and $Y_{2}$ or $Y_{3}$ respectively. 
\end{remark}

\subsubsection{Assessment of the MCAR, NMAR and MAR mechanisms}
 
It can be shown that perfect fits for fully and partially observed counts occur under Model C1 and not under Models C2, C3 and C4 (see Ghosh and Vellaisamy (2016b)). This implies the MLE's are $\widehat{\pi}_{ijk1} = y_{ijk1}/N$ and $\widehat{\pi}_{+jk2} = y_{+jk2}/N$. Hence, the estimators of the various response and non-response odds under Model C1 are as follows.
\begin{eqnarray*}
\widehat{\nu}_{ik}(j,j') &=& \frac{y_{ijk1}}{y_{ij'k1}},\widehat{\nu}_{nk}(j,j') = \min_{i}\{\widehat{\nu}_{ik}(j,j')\},\widehat{\nu}_{mk}(j,j') = \max_{i}\{\widehat{\nu}_{ik}(j,j')\},\widehat{\nu}_{k}(j,j') = \frac{y_{+jk2}}{y_{+j'k2}}; \nonumber \\
\widehat{\nu}_{ij}(k,k') &=& \frac{y_{ijk1}}{y_{ijk'1}},\widehat{\nu}_{nj}(k,k') = \min_{i}\{\widehat{\nu}_{ij}(k,k')\},\widehat{\nu}_{mj}(k,k') = \max_{i}\{\widehat{\nu}_{ij}(k,k')\},\widehat{\nu}_{j}(k,k') = \frac{y_{+jk2}}{y_{+jk'2}}. \nonumber 
\end{eqnarray*}
The MLE's of the response and the non-response odds under non-perfect fit models are some functions of those under perfect fit models. For example, the estimated cell probabilities under Model C4 (see Ghosh and Vellaisamy (2016b)) are
\begin{equation*}
\hat{\pi}_{ijk1} = \frac{y_{ijk1}y_{+jk+}y_{+++1}}{Ny_{+jk1}y_{++++}},\quad\hat{\pi}_{+jk2} = \frac{y_{+++2}\sum_{i}\hat{\pi}_{ijk1}}{y_{+++1}} = \frac{y_{+++2}y_{+jk+}}{Ny_{++++}}.
\end{equation*}
Hence, the MLE's of the above odds under Model C4 are
\begin{eqnarray*}
\widehat{\nu}_{ik}(j,j') &=& \frac{y_{ijk1}y_{+jk+}y_{+j'k1}}{y_{ij'k1}y_{+j'k+}y_{+jk1}},\widehat{\nu}_{nk}(j,j') = \min_{i}\{\widehat{\nu}_{ik}(j,j')\},\widehat{\nu}_{mk}(j,j') = \max_{i}\{\widehat{\nu}_{ik}(j,j')\}, \nonumber \\
\widehat{\nu}_{k}(j,j') &=& \frac{\sum_{i}\hat{\pi}_{ijk1}}{\sum_{i}\hat{\pi}_{ij'k1}} = \frac{y_{+jk+}}{y_{+j'k+}}; \nonumber \\
\widehat{\nu}_{ij}(k,k') &=& \frac{y_{ijk1}y_{+jk+}y_{+jk'1}}{y_{ijk'1}y_{+jk'+}y_{+jk1}},\widehat{\nu}_{nj}(k,k') = \min_{i}\{\widehat{\nu}_{ij}(k,k')\},\widehat{\nu}_{mj}(k,k') = \max_{i}\{\widehat{\nu}_{ij}(k,k')\}, \nonumber \\
\widehat{\nu}_{j}(k,k') &=& \frac{\sum_{i}\hat{\pi}_{ijk1}}{\sum_{i}\hat{\pi}_{ijk'1}} = \frac{y_{+jk+}}{y_{+jk'+}}. \nonumber 
\end{eqnarray*}
The MLE's of the various odds for Models C2 and C3 can be obtained similarly. Denote the estimators of $OI_{k}(j,j')$ and $OI_{j}(k,k')$ by $\widehat{OI}_{k}(j,j') = (\widehat{\nu}_{nk}(j,j'),~\widehat{\nu}_{mk}(j,j'))$ and $\widehat{OI}_{j}(k,k') = (\widehat{\nu}_{nj}(k,k'),~\widehat{\nu}_{mj}(k,k'))$ respectively. Then the corollary below follows from Theorem \ref{th3.1} and Remark \ref{rem23}.
\begin{corollary}\label{cor3.1}
For an $I\times J\times K\times 2$ table, if there exists at least one $k$ or at least one pair $(j,j')$ of $Y_{2}$ such that $\widehat{\nu}_{k}(j,j')\not\in \widehat{OI}_{k}(j,j')$, or if there exists at least one $j$ or at least one pair $(k,k')$ of $Y_{3}$ such that $\widehat{\nu}_{j}(k,k')\not\in \widehat{OI}_{j}(k,k')$, then the missing data mechanism of $Y_{1}$ is MAR, but neither NMAR nor MCAR. 
\end{corollary}

\subsection{Case 2: Missing in two of the variables} 

WLOG, suppose data on $Y_{1}$ and $Y_{2}$ are missing and for $i=1,2,$ let $R_{i}$ denote the missing indicator for $Y_{i}$ such that $R_{i} = 1$ if $Y_{i}$ is observed and $R_{i} = 2$ otherwise. Then for $Y_{1},~Y_{2},~Y_{3},~R_{1}$ and $R_{2},$ we have an $I\times J\times K\times 2\times 2$ table with cell counts ${\bf y} = \{y_{ijkrs}\}$ where $1\leq i\leq I,~1\leq j\leq J,~1\leq k\leq K$ and $r,s = 1,2.$ The vector of observed counts is ${\bf y_{\textrm{obs}}} = (\{y_{ijk11}\},\{y_{+jk21}\},\{y_{i+k12}\},\{y_{++k22}\}).$ For $I=J=K=2$, the $2\times 2\times 2\times 2\times 2$ incomplete table is given by Table \ref{t5}. 
\begin{table}[ht]
\caption{$2\times 2\times 2\times 2\times 2$ Incomplete Table}\label{t5}
\begin{center}
$
\begin{array}{|c|c|cc|cc|}\hline
& & & & Y_{3} = 1 & Y_{3} = 2 \\ \hline
R_{1} = 1 & Y_{1} = 1 & R_{2} = 1 & Y_{2} = 1 & y_{11111} & y_{11211} \\   
& & & Y_{2} = 2 & y_{12111} & y_{12211} \\ \hline 
& & R_{2} = 2 & \text{Missing} & y_{1+112} & y_{1+212} \\ \hline
& Y_{1} = 2 & R_{2} = 1 & Y_{2} = 1 & y_{21111} & y_{21211} \\   
& & & Y_{2} = 2 & y_{22111} & y_{22211} \\ \hline 
& & R_{2} = 2 & \text{Missing} & y_{2+112} & y_{2+212} \\ \hline
R_{1} = 2 & \text{Missing} & R_{2} = 1 & Y_{2} = 1 & y_{+1121} & y_{+1221} \\   
& & & Y_{2} = 2 & y_{+2121} & y_{+2221} \\ \hline 
& & R_{2} = 2 & \text{Missing} & y_{++122} & y_{++222} \\ \hline  
\end{array}
$
\end{center}
\end{table}
Let ${\bf\pi} = \{\pi_{ijkrs}\}$ be the vector of cell probabilities, $\mu = \{\mu_{ijkrs}\}$ be the vector of expected counts and $N$ be the total cell count. Under Poisson sampling, the log-likelihood kernel of ${\bf\mu}$ is 
\begin{eqnarray}\label{eq3.2}
l({\bf\mu};{\bf y}_{\textrm{obs}}) &=& \sum_{i,j,k}y_{ijk11}\log \mu_{ijk11} + \sum_{j,k}y_{+jk21}\log \mu_{+jk21} + \sum_{i,k}y_{i+k12}\log \mu_{i+k12} \nonumber \\ 
& & + \sum_{k}y_{++k22}\log \pi_{++k22} - \sum_{i,j,k,r,s}\mu_{ijkrs}. \nonumber
\end{eqnarray}
The log-linear model in this case is (see Ghosh and Vellaisamy (2016a)) 
\begin{eqnarray}\label{eq3}
\log \mu_{ijkrs} &=& \lambda + \lambda_{Y_{1}}(i) + \lambda_{Y_{2}}(j) + \lambda_{Y_{3}}(k) + \lambda_{R_{1}}(x) + \lambda_{R_{2}}(s) + \lambda_{Y_{1}Y_{2}}(i,j) \nonumber \\
& & + \lambda_{Y_{1}Y_{3}}(i,k) + \lambda_{Y_{2}Y_{3}}(j,k) + \lambda_{Y_{1}R_{1}}(i,r) + \lambda_{Y_{2}R_{1}}(j,r) + \lambda_{Y_{3}R_{1}}(k,r) \nonumber \\
& & + \lambda_{Y_{1}R_{2}}(i,s) + \lambda_{Y_{2}R_{2}}(j,s) + \lambda_{Y_{3}R_{2}}(k,s) + \lambda_{R_{1}R_{2}}(r,s). 
\end{eqnarray}
Three-way and higher order associations are assumed to be zero in (\ref{eq3}) as they are difficult to interpret. Also, analysis by ML estimation (without using iterative procedures) becomes intractable requiring more parameters, which might lead to non-identifiable models. Rochani {\it et al.}~(2017) used the term `homogeneous log-linear model' to denote a log-linear model that permits all two-way interactions between variables as well as missing-data parameters but contains no higher-order interactions, for example, the model (\ref{eq3}).

Note that association terms among $Y_i$'s and those among $R_i$'s are not involved in studying the missing data mechanisms of $Y_i$'s in (\ref{eq3}). Hence, there is no need to include three-way or higher order interactions among the outcome variables such as $\lambda_{Y_{1}Y_{2}Y_{3}}$ or the missing indicators such as $\lambda_{R_{1}R_{2}R_{3}}$. It is assumed that the MAR mechanism of a variable depends on any one of the other variables so that interaction terms like $\lambda_{Y_{i}Y_{j}R_{k}}$ for $i\neq j\neq k$ are excluded from (\ref{eq3}). Also, the missingness mechanism of a variable cannot be NMAR and MAR simultaneously, which excludes terms with $\lambda_{Y_{i}Y_{j}R_{i}}$ for $i \neq j$ in (\ref{eq3}). Interactions such as $\lambda_{Y_{i}R_{k}R_{l}}$ for $i\neq k\neq l$ are absent in (\ref{eq3}) since their interpretation is unclear. Also, they are redundant for the derivation of closed-form estimates of the expected cell counts. Each log-linear parameter in (\ref{eq3}) satisfies the constraint that the sum over each of its arguments is 0. Based on the assumption regarding various missing mechanisms of a variable in Case 1, there are 16 identifiable missing data models, which are submodels of (\ref{eq3}) and categorized as follows.
\begin{enumerate}
\item[D1.] MCAR model for both $Y_{1}$ and $Y_{2}$ (1 model), 
\item[D2.] NMAR model for both $Y_{1}$ and $Y_{2}$ (1 model), 
\item[D3.] MAR models for both $Y_{1}$ and $Y_{2}$ (4 models),  
\item[D4.] Mixture of MCAR and NMAR models for $Y_{1}$ and $Y_{2}$ (2 models), 
\item[D5.] Mixture of MCAR and MAR models for $Y_{1}$ and $Y_{2}$ (4 models), 
\item[D6.] Mixture of NMAR and MAR models for $Y_{1}$ and $Y_{2}$ (4 models).
\end{enumerate}

\subsubsection{Properties of the missing data models} 

Define the odds $\nu'_{ij}(k,k')$, $\nu'_{mj}(k,k')$, $\nu'_{nj}(k,k')$, $\nu'_{j}(k,k')$, $\nu'_{ik}(j,j')$, $\nu'_{mk}(j,j')$, $\nu'_{nk}(j,j')$ and $\nu'_{k}(j,j')$ similarly as the corresponding ones defined for the case when $Y_{1}$ is missing in an $I\times J\times K\times 2$ table. In this case, replace $\pi_{ijk1}$ by $\pi_{ijk11}$, $\pi_{+jk2}$ by $\pi_{+jk21}$, $\pi_{ij'k1}$ by $\pi_{ij'k11}$, $\pi_{+j'k2}$ by $\pi_{+j'k21}$, $\pi_{ijk'1}$ by $\pi_{ijk'11}$ and $\pi_{+jk'2}$ by $\pi_{+jk'21}$. Also, define the following response and non-response odds based on $\pi$. 
\begin{equation*}
\omega'_{jk}(i,i') = \frac{\pi_{ijk11}}{\pi_{i'jk11}},~\omega'_{nk}(i,i') = \min_{j}\{\omega'_{jk}(i,i')\},~\omega'_{mk}(i,i') = \max_{j}\{\omega'_{jk}(i,i')\},~\omega'_{k}(i,i') = \frac{\pi_{i+k12}}{\pi_{i'+k12}};
\end{equation*} 
\begin{equation*}
\omega'_{ji}(k,k') = \frac{\pi_{ijk11}}{\pi_{ijk'11}},\omega'_{ni}(k,k') = \min_{j}\{\omega'_{ji}(k,k')\},\omega'_{mi}(k,k') = \max_{j}\{\omega'_{ji}(k,k')\},\omega'_{i}(k,k') = \frac{\pi_{i+k12}}{\pi_{i+k'12}}.
\end{equation*}
Let $OI'_{k}(i,i') = (\omega'_{nk}(i,i'),~\omega'_{mk}(i,i'))$, $OI'_{i}(k,k') = (\omega'_{ni}(k,k'),~\omega'_{mi}(k,k'))$, \\ $OI'_{k}(j,j') = (\nu'_{nk}(j,j'),~\nu'_{mk}(j,j'))$ and $OI'_{j}(k,k') = (\nu'_{nj}(k,k'),~\nu'_{mj}(k,k'))$. Applying the methods described for $I\times J\times 2\times 2$ and $I\times J\times K\times 2$ tables, the conditions under which the non-response odds belong to the open intervals formed by the response odds for models in D1-D6 may be obtained. Specifically, the following inequalities hold under D3.   
\begin{eqnarray}
A'_{mk}(j,j') &=& \frac{\sum_{i}\exp\{\lambda_{Y_{1}}(i) + \lambda_{Y_{1}Y_{2}}(i,j') + \lambda_{Y_{1}Y_{2}}(m,j) + \lambda_{Y_{1}Y_{3}}(i,k) + \lambda_{Y_{1}R_{2}}(i,1)\}}{\sum_{i}\exp\{\lambda_{Y_{1}}(i) + \lambda_{Y_{1}Y_{2}}(i,j) + \lambda_{Y_{1}Y_{2}}(m,j') + \lambda_{Y_{1}Y_{3}}(i,k) + \lambda_{Y_{1}R_{2}}(i,1)\}} > 1, \nonumber \\
& & \label{c1} \\
B'_{mk}(i,i') &=& \frac{\sum_{j}\exp\{\lambda_{Y_{2}}(j) + \lambda_{Y_{1}Y_{2}}(i',j) + \lambda_{Y_{1}Y_{2}}(i,m) + \lambda_{Y_{2}Y_{3}}(j,k) + \lambda_{Y_{2}R_{1}}(j,1)\}}{\sum_{j}\exp\{\lambda_{Y_{2}}(j) + \lambda_{Y_{1}Y_{2}}(i,j) + \lambda_{Y_{1}Y_{2}}(i',m) + \lambda_{Y_{2}Y_{3}}(j,k) + \lambda_{Y_{2}R_{1}}(j,1)\}} > 1, \nonumber \\
& & \label{c2}
\end{eqnarray}
where $m$ denotes the level of $Y_{1}$ ($Y_{2}$) corresponding to $\nu'_{mk}(j,j')$ ($\omega'_{mk}(i,i')$) for $A'_{mk}(j,j')$ ($B'_{mk}(i,i')$). Similarly, we may obtain $A'_{mj}(k,k') > 1$ by interchanging $j$ with $k$ in (\ref{c1}) and $B'_{mi}(k,k') > 1$ by interchanging $i$ with $k$ in (\ref{c2}). Also, we have $A'_{nk}(j,j'),A'_{nj}(k,k'),B'_{ni}(k,k')$ and $B'_{nk}(i,i')$, each less than 1, by replacing $m$ with $n$. For D5 and D6, some of the above inequalities hold with $\lambda_{Y_{2}R_{1}}(j,1) = 0$ in $B'_{mi}(k,k')$, $B'_{mk}(i,i')$, $B'_{ni}(k,k')$ and $B'_{nk}(i,i')$ or $\lambda_{Y_{1}R_{2}}(i,1) = 0$ in $A'_{mk}(j,j')$, $A'_{mj}(k,k')$, $A'_{nk}(j,j')$ and $A'_{nj}(k,k')$. Denote $A^{\ast}$ to be $A^{+}$ ($A'$ with $\lambda_{Y_{1}R_{2}}(i,1)\neq 0$) or $A^{-}$ ($A'$ with $\lambda_{Y_{1}R_{2}}(i,1) = 0$). Also, denote $B^{\ast}$ to be $B^{+}$ ($B'$ with $\lambda_{Y_{2}R_{1}}(j,1)\neq 0$) or $B^{-}$ ($B'$ with $\lambda_{Y_{2}R_{1}}(j,1) = 0$). Then the next result provides characterizations for the missing mechanisms of a variable in an $I\times J\times K\times 2\times 2$ table. This result aids in the assessment of the MCAR, NMAR and MAR assumptions of the above models based on only the observed cell counts or their sums. 
\begin{theorem}\label{th4.1}
Under the missing data models in D1-D6 for an $I\times J\times K\times 2\times 2$ table, we have the following cases corresponding to the missing mechanism of $Y_{1}$. \\
(a) If $Y_{1}$ has a NMAR or MCAR mechanism, then 
\begin{enumerate}
\item[(i)] $\nu'_{k}(j,j')\in OI'_{k}(j,j')$ if $|\lambda_{Y_{1}R_{1}}(i,2)| < \infty$, 
\item[(ii)] $\nu'_{j}(k,k')\in OI'_{k}(k,k')$ if $|\lambda_{Y_{1}R_{1}}(i,2)| < \infty$. 
\end{enumerate}
(b) If $Y_{1}$ has a MAR mechanism, then only one of the Conditions (1a,1b) and (2a,2b) holds: 
\begin{enumerate}
\item[1a.] For each $k$ and each pair $(j,j')$ of $Y_{2}$, only one of the conditions below holds: \\
(i) $\nu'_{k}(j,j')\in OI'_{k}(j,j')$ iff $-\frac{1}{2}\log A^{\ast}_{mk}(j,j') < \lambda_{Y_{2}R_{1}}(j',2) - \lambda_{Y_{2}R_{1}}(j,2) < -\frac{1}{2}\log A^{\ast}_{nk}(j,j')$, \\
(ii) $\nu'_{k}(j,j')\not\in OI'_{k}(j,j')$ iff $\lambda_{Y_{2}R_{1}}(j',2) - \lambda_{Y_{2}R_{1}}(j,2) > -\frac{1}{2}\log A^{\ast}_{nk}(j,j')$ or $\lambda_{Y_{2}R_{1}}(j',2) - \lambda_{Y_{2}R_{1}}(j,2) < -\frac{1}{2}\log A^{\ast}_{mk}(j,j')$,  
\item[1b.] $\nu'_{j}(k,k')\in OI'_{j}(k,k')$; 
\item[2a.] $\nu'_{k}(j,j')\in OI'_{k}(j,j')$,   
\item[2b.] For each $j$ and each pair $(k,k')$ of $Y_{3}$, only one of the conditions below holds: \\
(i) $\nu'_{j}(k,k')\in OI'_{j}(k,k')$ iff $-\frac{1}{2}\log A^{\ast}_{mj}(k,k') < \lambda_{Y_{3}R_{1}}(k',2) - \lambda_{Y_{3}R_{1}}(k,2) < -\frac{1}{2}\log A^{\ast}_{nj}(k,k')$, \\
(ii) $\nu'_{j}(k,k')\not\in OI'_{j}(k,k')$ iff $\lambda_{Y_{3}R_{1}}(k',2) - \lambda_{Y_{3}R_{1}}(k,2) > -\frac{1}{2}\log A^{\ast}_{nj}(k,k')$ or $\lambda_{Y_{3}R_{1}}(k',2) - \lambda_{Y_{3}R_{1}}(k,2) < -\frac{1}{2}\log A^{\ast}_{mj}(k,k')$.
\end{enumerate}
\end{theorem}
\begin{proof}
See Appendix E.
\end{proof}
A similar result for the various missing mechanisms of $Y_{2}$ under models in D1-D6 for an $I\times J\times K\times 2\times 2$ table can be obtained. 

\subsubsection{Assessment of the MCAR, NMAR and MAR mechanisms}

Here, we propose a method to assess the MCAR, NMAR and MAR  assumptions in an $I\times J\times K\times 2\times 2$ table. It can be shown that perfect fits for fully and partially observed data occur for models in D2 and D6 (see Ghosh and Vellaisamy (2016b) or Rochani {\it et al.}~(2017)). So, the MLE's of $\pi_{ijk11}$, $\pi_{i+k12}$ and $\pi_{+jk21}$ under the above models are $y_{ijk11}/N$, $y_{i+k12}/N$ and $y_{+jk21}/N$ respectively. This implies that the estimators $\widehat{\nu}'_{ij}(k,k')$, $\widehat{\nu}'_{mj}(k,k')$, $\widehat{\nu}'_{nj}(k,k')$, $\widehat{\nu}'_{j}(k,k')$, $\widehat{\nu}'_{ik}(j,j')$, $\widehat{\nu}'_{mk}(j,j')$, $\widehat{\nu}'_{nk}(j,j')$ and $\widehat{\nu}'_{k}(j,j')$ are similar to the corresponding ones defined for the case when $Y_{1}$ is missing in an $I\times J\times K\times 2$ table. In this case, replace $y_{ijk1}$ by $y_{ijk11}$, $y_{+jk2}$ by $y_{+jk21}$, $y_{ij'k1}$ by $y_{ij'k11}$, $y_{+j'k2}$ by $y_{+j'k21}$, $y_{ijk'1}$ by $y_{ijk'11}$ and $y_{+jk'2}$ by $y_{+jk'21}$. Also, the estimators of the odds $\omega$'s defined in the previous subsection are given below. 
\begin{equation*}
\widehat{\omega}'_{jk}(i,i') = \frac{y_{ijk11}}{y_{i'jk11}},~\widehat{\omega}'_{nk}(i,i') = \min_{j}\{\widehat{\omega}'_{jk}(i,i')\},~\widehat{\omega}'_{mk}(i,i') = \max_{j}\{\widehat{\omega}'_{jk}(i,i')\},~\widehat{\omega}'_{k}(i,i') = \frac{y_{i+k12}}{y_{i'+k12}};
\end{equation*} 
\begin{equation*}
\widehat{\omega}'_{ji}(k,k') = \frac{y_{ijk11}}{y_{ijk'11}},\widehat{\omega}'_{ni}(k,k') = \min_{j}\{\widehat{\omega}'_{ji}(k,k')\},\widehat{\omega}'_{mi}(k,k') = \max_{j}\{\widehat{\omega}'_{ji}(k,k')\},\widehat{\omega}'_{i}(k,k') = \frac{y_{i+k12}}{y_{i+k'12}}
\end{equation*}
The estimated expected counts and hence the response and the non-response odds under non-perfect fit models are functions of those under perfect-fit models. For example, the MLE's of the cell probabilities under a model in D4 (MCAR for $Y_{1}$, NMAR for $Y_{2}$) are (see Ghosh and Vellaisamy (2016b) or Rochani {\it et al.}~(2017))
\begin{equation*}
\hat{\pi}_{ijk11} = \frac{y_{ijk11}y_{+++11}y_{+jk+1}}{Ny_{++++1}y_{+jk11}},\quad\hat{\pi}_{+jk21} = \frac{y_{+++21}\sum_{i}\hat{\pi}_{ijk11}}{y_{+++11}} = \frac{y_{+jk+1}y_{+++21}}{Ny_{+++1}},\quad\hat{\pi}_{i+k12} = \frac{y_{i+k12}}{N}.
\end{equation*}
So the MLE's of the odds under the above model are
\begin{eqnarray*}
\widehat{\nu}'_{ik}(j,j') &=& \frac{y_{ijk11}y_{+jk+1}y_{+j'k11}}{y_{ij'k11}y_{+j'k+1}y_{+jk11}}, \widehat{\nu}'_{nk}(j,j') = \min_{i}\{\widehat{\nu}'_{ik}(j,j')\}, \widehat{\nu}'_{mk}(j,j') = \max_{i}\{\widehat{\nu}'_{ik}(j,j')\}, \nonumber \\
\widehat{\nu}'_{k}(j,j') &=& \frac{y_{+jk+1}}{y_{+j'k+1}}; \nonumber \\
\widehat{\nu}'_{ij}(k,k') &=& \frac{y_{ijk11}y_{+jk+1}y_{+jk'11}}{y_{ijk'11}y_{+jk'+1}y_{+jk11}}, \widehat{\nu}'_{nj}(k,k') = \min_{i}\{\widehat{\nu}'_{ij}(k,k')\}, \widehat{\nu}'_{mj}(k,k') = \max_{i}\{\widehat{\nu}'_{ij}(k,k')\}, \nonumber \\
\widehat{\nu}'_{j}(k,k') &=& \frac{y_{+jk+1}}{y_{+jk'+1}}; \nonumber \\
\widehat{\omega}'_{jk}(i,i') &=& \frac{y_{ijk11}}{y_{i'jk11}},~\widehat{\omega}'_{nk}(i,i') = \min_{j}\{\widehat{\omega}'_{jk}(i,i')\},~\widehat{\omega}'_{mk}(i,i') = \max_{j}\{\widehat{\omega}'_{jk}(i,i')\},~\widehat{\omega}'_{k}(i,i') = \frac{y_{i+k12}}{y_{i'+k12}}; \nonumber \\
\widehat{\omega}'_{ji}(k,k') &=& \frac{y_{ijk11}y_{+jk+1}y_{+jk'11}}{y_{ijk'11}y_{+jk'+1}y_{+jk11}},\widehat{\omega}'_{ni}(k,k') = \min_{j}\{\widehat{\omega}'_{ji}(k,k')\},\widehat{\omega}'_{mi}(k,k') = \max_{j}\{\widehat{\omega}'_{ji}(k,k')\}, \nonumber \\
\widehat{\omega}'_{i}(k,k') &=& \frac{y_{i+k12}}{y_{i+k'12}}. \nonumber
\end{eqnarray*} 
The MLE's of the response and non-response odds under the other non-perfect fit models can similarly be obtained using the corresponding estimated expected counts. Let $\widehat{OI}'_{k}(i,i') = (\widehat{\omega}'_{nk}(i,i'),~\widehat{\omega}'_{mk}(i,i'))$, $\widehat{OI}'_{i}(k,k') = (\widehat{\omega}'_{ni}(k,k'),~\widehat{\omega}'_{mi}(k,k'))$, $\widehat{OI}'_{k}(j,j') = (\widehat{\nu}'_{nk}(j,j'),~\widehat{\nu}'_{mk}(j,j'))$ and $\widehat{OI}'_{j}(k,k') = (\widehat{\nu}'_{nj}(k,k'),~\widehat{\nu}'_{mj}(k,k'))$. Then we have the following corollary based on Theorem \ref{th4.1}(a) and a similar result for the various missing mechanisms of $Y_{2}$. 
\begin{corollary}\label{cor3.2}
For an $I\times J\times K\times 2\times 2$ table, if $\widehat{\nu}'_{k}(j,j')\not\in \widehat{OI}'_{k}(j,j')$ or $\widehat{\nu}'_{k}(j,j')\not\in \widehat{OI}'_{k}(j,j')$ or $\widehat{\omega}'_{k}(i,i')\not\in \widehat{OI}'_{k}(i,i')$ or $\widehat{\omega}'_{i}(k,k')\not\in \widehat{OI}'_{i}(k,k')$ for at least one of $i,j,k$ or one of the pairs $(i,i'),(j,j'),(k,k')$, then the plausible missing mechanism of $Y_{1}$ or $Y_{2}$ is MAR, but neither NMAR nor MCAR.
\end{corollary}

\subsection{Case 3: Missing in all three variables}
 
We omit the details for this case. Similar to the tables discussed earlier, we can characterize the missing mechanisms of a variable for an $I\times J\times K\times 2\times 2\times 2$ table. Also, a method to assess the MCAR, MAR and NMAR assumptions may be obtained as in Cases 1 and 2 using estimators of the response and non-response odds, and open intervals, which depend only on the observed cell counts.

\begin{remark}\label{rem33}
It is clear from the results in Sections 2 and 3 what the proposed methods in our paper can achieve and cannot. We have provided characterizations (see Theorems \ref{th2.2}, \ref{th2.1} and \ref{th3.1}-\ref{th4.1}) for each missing mechanism (NMAR, MCAR and MAR) of a variable in an incomplete table. However, our proposed methods (see Corollaries \ref{cor2.1}, \ref{cor3.1} and \ref{cor3.2}) can mainly verify if the missing mechanism of at least one of the variables in an incomplete table is MAR (neither NMAR nor MCAR) based on only the observed cell counts, but cannot identify NMAR and MCAR mechanisms separately. 
\end{remark}

\section{Data analysis}

In this section, we analyze some real-life datasets to demonstrate our results (Corollaries \ref{cor2.1}, \ref{cor3.1} and \ref{cor3.2}) on the assessment of missing data mechanisms in Sections 2 and 3. Note that the results are robust with respect to some minor changes in the data as long as any of the given conditions in the corollaries is satisfied.
\begin{example}\label{ex1}
Consider the data in Table \ref{t6} discussed in Baker {\it et al.} (1992), which cross-classifies mother's self-reported smoking status ($Y_{1}$) ($Y_{1} = 1(2)$ for smoker (non-smoker)) with newborn's weight ($Y_{2}$) ($Y_{2} = 1(2)$ if weight $< 2500$ grams ($\geq 2500$ grams)). The fully observed data comprises data on both smoking data and newborn's weight ($R_{1}=R_{2}=1$), while the supplementary margins contain data on only smoking status ($R_{1}=1,R_{2}=2$), data on only newborn's weight ($R_{1}=2,R_{2}=1$) and missing data on both smoking data and newborn's weight ($R_{1}=R_{2}=2$) . Comparing with Table \ref{t1}, we observe from Table \ref{t6} that $y_{1111}=4512$, $y_{1211}=21009$, $y_{2111}=3394$, $y_{2211}=24132$, $y_{1+12}=1049$, $y_{2+12}=1135$, $y_{+121}=142$, $y_{+221}=464$ and $y_{++22}=1224$, which help us compute the MLE's of the various odds (see Section 2.2) under Models M1-M8. 
\begin{table}[ht]   
\caption{Birth weight and smoking : observed counts}\label{t6}
\begin{center}
$
\begin{array}{|c|c|cc|c|}\hline
& & R_{2} = 1 & & R_{2} = 2 \\ \hline
& & Y_{2}=1 & Y_{2}=2 & Y_{2}~\textrm{missing} \\ \hline  
R_{1} = 1 & Y_{1}=1 & 4512 & 21009 & 1049 \\
& Y_{1}=2 & 3394 & 24132 & 1135 \\ \hline
R_{1} = 2 & Y_{1}~\textrm{missing} & 142 & 464 & 1224 \\ \hline
\end{array}
$
\end{center}
\end{table}
 Note that from Table \ref{t6}, the MLE's satisfy $\frac{142}{464}\not\in\left(\frac{3394}{24132},\frac{4512}{21009}\right)$, while $\frac{1049}{1135}\in\left(\frac{21009}{24132},\frac{4512}{3394}\right)$. Hence, from Corollary \ref{cor2.1}, the missing data mechanism of $Y_{1}$ or $Y_{2}$ is likely MAR, but neither NMAR nor MCAR. This result coincides with the analysis by Baker {\it et al.} (1992) who infer that the most parsimonious fit model is MAR for $Y_{1}$ and MCAR for $Y_{2}$ (Model M4 in Section 2). They also mention that boundary solutions occur on fitting NMAR models for $Y_{1}$ (Models M1, M2 and M3 in Section 2) to the dataset in Table \ref{t6}. This implies that these models provide poor fits to the observed data (see Clarke and Smith (2005)), which further supports our observation. Based on $p$-value of 1, the other plausible models are M5 and M6 (see Section 2) for which the missing mechanism is MAR for $Y_{1}$.  

To assess the uncertainty of the accuracy of the proposed method, bootstrap resampling is performed. We generate 10,000 random samples from the Models M4, M5 and M6 fitted to the data and check for those samples satisfying the condition on $\widehat{\nu}$'s in Corollary \ref{cor2.1}. The computed percentage of such samples is 99.99 for each of the above models, which confirms the accuracy of the proposed method. 
\end{example}

\begin{example}\label{ex2}
Consider Table \ref{t7} discussed in Park {\it et al.} (2014), which cross-classifies data on bone mineral density ($Y_{1}$) and family income ($Y_{2}$) in a $3\times 3\times 2\times 2$ incomplete table. Both variables $Y_{1}$ and $Y_{2}$ have three levels. The total count is $2998$ out of which data on $Y_{1}$ and $Y_{2}$ are available ($R_{1}=R_{2}=1$) for $1844$ persons, data on $Y_{1}$ only ($R_{1}=1,R_{2}=2$) for $231$ persons, data on $Y_{2}$ only ($R_{1}=2,R_{2}=1$) for $878$ persons, and data on neither of them ($R_{1}=R_{2}=2$)  for $45$ persons. Specifically from Table \ref{t7}, we have $y_{1111}=621$, $y_{1211}=290$, $y_{1311}=284$, $y_{2111}=260$, $y_{2211}=131$, $y_{2311}=117$, $y_{3111}=93$, $y_{3211}=30$, $y_{3311}=18$, $y_{1+12}=135$, $y_{2+12}=69$, $y_{3+12}=27$, $y_{+121}=456$, $y_{+221}=156$, $y_{+321}=266$ and $y_{++22}=45$, using which the MLE's of the various odds (see Section 2.2) can be obtained.
\begin{table}[ht]
\caption{Bone mineral density ($Y_{1}$) and family income ($Y_{2}$)}\label{t7}
\begin{center}
$
\begin{array}{|c|c|ccc|c|}\hline
& & R_{2} = 1 & & & R_{2} = 2 \\
& & Y_{2} = 1 & Y_{2} = 2 & Y_{2} = 3 & \text{Missing} \\ \hline
& Y_{1} = 1 & 621 & 290 & 284 & 135 \\   
R_{1} = 1 & Y_{1} = 2 & 260 & 131 & 117 & 69 \\
& Y_{1} = 3 & 93 & 30 & 18 & 27 \\ \hline
R_{1} = 2 & \text{Missing} & 456 & 156 & 266 & 45 \\ \hline
\end{array}
$
\end{center}
\end{table}
From Table \ref{t7}, note that the MLE's satisfy $\frac{456}{156}\in\left(\frac{260}{131},\frac{93}{30}\right)$, $\frac{456}{266}\not\in\left(\frac{621}{284},\frac{93}{18}\right)$ and $\frac{156}{266}\not\in\left(\frac{290}{284},\frac{30}{18}\right)$, while $\frac{135}{69}\not\in\left(\frac{290}{131},\frac{284}{117}\right)$, $\frac{135}{27}\not\in\left(\frac{621}{93},\frac{284}{18}\right)$ and $\frac{69}{27}\not\in\left(\frac{260}{93},\frac{117}{18}\right)$. 

Hence, it follows from Corollary \ref{cor2.1} that the missing mechanism of $Y_{1}$ or $Y_{2}$ is probably MAR, but neither NMAR nor MCAR. Let $G^{2}$ denote the likelihood ratio statistic for testing the goodness of fit of the proposed model against the perfect fit model. When we fit Models M1-M9 (see Section 2) to the data in Table \ref{t7}, we deduce that the plausible models are M4 (MAR for $Y_{1}$, MCAR for $Y_{2}$) and M5 (MAR for both $Y_{1}$ and $Y_{2}$) based on $p$-values of 0.066 and 1, and $G^{2}$ values of 5.42 and 0 respectively. This analysis supports our earlier observation. Park {\it et al.} (2014) also showed that boundary solutions occur on fitting NMAR models for $Y_{1}$ or $Y_{2}$ (Models M1, M2, M3, M6 and M8), which thereby fit the given data poorly and hence provides further evidence for our result. 

We carry out bootstrap resampling to evaluate the performance of the proposed method. We generate 10,000 random samples from the Models M4 and M5 fitted to the data and verify those samples satisfying the condition on $\widehat{\nu}$'s or on $\widehat{\omega}$'s in Corollary \ref{cor2.1}. The computed percentage of such samples for $\widehat{\nu}$'s ($\widehat{\omega}$'s) is 99.99  (69.06) and 99.99 (96.56) under models M4 and M5 respectively, which confirms the accuracy of the proposed method. 
\end{example}   

\begin{example}\label{ex3}
Here, we consider a real-life example from Rubin {\it et al.} (1995). Based on the Slovenian public opinion (SPO) survey, the dataset shown in Table \ref{t8} is a $2\times 2\times 2\times 2\times 2\times 2$ table classified by the variables Secession ($Y_{1}$), Attendance ($Y_{2}$) and Independence ($Y_{3}$), each having two levels Yes (1) and No (2). The ``Don't know" category (missing margins) for each variable is denoted by ``Missing". Since $G^{2}$ becomes undefined for null cell count, we replace the count 0 by 2 in the full table. The total cell count is 2076, consisting of data on all three variables observed ($R_{1} = R_{2} = R_{3} = 1$) for 1456 persons, $Y_{1}$ and $Y_{2}$ observed ($R_{1} = R_{2} = 1, R_{3} = 2$) for 57 persons, $Y_{1}$ and $Y_{3}$ observed ($R_{1} = R_{3} = 1, R_{2} = 2$) for 171 persons, $Y_{2}$ and $Y_{3}$ observed ($R_{2} = R_{3} = 1, R_{1} = 2$) for 95 persons, only $Y_{1}$ observed ($R_{2} = R_{3} =2, R_{1} = 1$) for 40 persons, only $Y_{2}$ observed ($R_{1} = R_{3} = 2, R_{2} = 1$) for 134 persons, only $Y_{3}$ observed ($R_{1} = R_{2} = 2, R_{3} = 1$) for 27 persons, and all missing ($R_{1} = R_{2} = R_{3} = 2$) for 96 persons. 
\begin{table}[ht]
\caption{Data from the SPO survey}\label{t8}
\begin{center}
$
\begin{array}{|c|c|cc|cc|c|}\hline
& & & & R_{3} = 1 & & R_{3} = 2 \\
& & & & Y_{3} = 1 & Y_{3} = 2 & \text{Missing} \\ \hline
R_{1} = 1 & Y_{1} = 1 & R_{2} = 1 & Y_{2} = 1 & 1191 & 8 & 21 \\   
& & & Y_{2} = 2 & 8 & 2 & 4 \\ \hline 
& & R_{2} = 2 & \text{Missing} & 107 & 3 & 9 \\ \hline
& Y_{1} = 2 & R_{2} = 1 & Y_{2} = 1 & 158 & 68 & 29 \\   
& & & Y_{2} = 2 & 7 & 14 & 3 \\ \hline 
& & R_{2} = 2 & \text{Missing} & 18 & 43 & 31 \\ \hline
R_{1} = 2 & \text{Missing} & R_{2} = 1 & Y_{2} = 1 & 90 & 2 & 109 \\   
& & & Y_{2} = 2 & 1 & 2 & 25 \\ \hline 
& & R_{2} = 2 & \text{Missing} & 19 & 8 & 96 \\ \hline  
\end{array}
$
\end{center}
\end{table}
\vone WLOG, consider the subtable (Table \ref{t9}) of Table \ref{t8} in which data on only $Y_{1}$ is missing. 
\begin{table}[ht]
\caption{Subtable of Table \ref{t8} for $Y_{1}$}\label{t9}
\begin{center}
$
\begin{array}{|c|c|c|cc|}\hline
& & & Y_{3} = 1 & Y_{3} = 2 \\ \hline
R = 1 & Y_{1} = 1 & Y_{2} = 1 & 1191 & 8 \\   
& & Y_{2} = 2 & 8 & 2  \\ \hline 
& Y_{1} = 2 & Y_{2} = 1 & 158 & 68 \\   
& & Y_{2} = 2 & 7 & 14 \\ \hline
R = 2 & \text{Missing} & Y_{2} = 1 & 90 & 2 \\   
& & Y_{2} = 2 & 1 & 2 \\ \hline 
\end{array}
$
\end{center}
\end{table}
Comparing Table \ref{t9} with Table \ref{t4}, we have $y_{11111}=1191$, $y_{11211}=8$, $y_{12111}=8$, $y_{12211}=2$, $y_{21111}=158$, $y_{21211}=68$, $y_{22111}=7$, $y_{22211}=14$, $y_{+1121}=90$, $y_{+1221}=2$, $y_{+2121}=1$ and $y_{+2221}=2$ leading to MLE's of the various response and non-response odds (see Section 3.1.2). From Table \ref{t9}, we observe that these MLE's satisfy $\frac{90}{2}\in\left(\frac{158}{68},\frac{1191}{8}\right)$, $\frac{1}{2}\not\in\left(\frac{7}{14},\frac{8}{2}\right)$, $\frac{90}{1}\in\left(\frac{158}{7},\frac{1191}{8}\right)$ and $\frac{2}{2}\not\in\left(\frac{8}{2},
\frac{68}{14}\right)$. Hence, from Corollary \ref{cor3.1}, the plausible missing mechanism of $Y_{1}$ is MAR, but neither NMAR nor MCAR. Let $G^{2}$ denote the likelihood ratio statistic for testing the goodness of fit of the proposed model against the perfect fit model. On fitting Models C1-C4 (see Section 3.1) to the data in Table \ref{t9} and discarding the perfect fit model C1, the plausible models are Models C2, C3 and C4 based on $p$-values of 0.29, 0.35 and 0.41 respectively. However, we deduce that the best fit model is Model C3 (MAR for $Y_{1}$) based on minimum $G^{2}$ value of 2.0949. This observation is consistent with our earlier result (Corollary \ref{cor3.1}).  

To evaluate the uncertainty of the accuracy of the proposed method, bootstrap resampling is performed. We generate 10,000 random samples from the Models C2, C3 and C4 fitted to the data and count the number of samples satisfying the conditions on $\widehat{\nu}$'s in Corollary \ref{cor3.1}. The computed percentages of such samples are 88.99, 96.95 and 89.95 under Models C2, C3 and C4 respectively, which confirms the accuracy of the proposed method. 
\end{example}

\begin{example}\label{ex4}
In this example, we use the data in Table \ref{t8} of Example \ref{ex3}. Consider WLOG the subtable (Table \ref{t10}) of Table \ref{t8} in which data on $Y_{1}$ and $Y_{2}$ are missing. 
\begin{table}[ht]
\caption{Subtable of Table \ref{t8} for $Y_{1}$ and $Y_{2}$}\label{t10}
\begin{center}
$
\begin{array}{|c|c|cc|cc|}\hline
& & & & Y_{3} = 1 & Y_{3} = 2 \\ \hline
R_{1} = 1 & Y_{1} = 1 & R_{2} = 1 & Y_{2} = 1 & 1191 & 8 \\   
& & & Y_{2} = 2 & 8 & 2 \\ \hline 
& & R_{2} = 2 & \text{Missing} & 107 & 3 \\ \hline
& Y_{1} = 2 & R_{2} = 1 & Y_{2} = 1 & 158 & 68 \\   
& & & Y_{2} = 2 & 7 & 14 \\ \hline 
& & R_{2} = 2 & \text{Missing} & 18 & 43 \\ \hline
R_{1} = 2 & \text{Missing} & R_{2} = 1 & Y_{2} = 1 & 90 & 2 \\   
& & & Y_{2} = 2 & 1 & 2 \\ \hline 
& & R_{2} = 2 & \text{Missing} & 19 & 8 \\ \hline  
\end{array}
$
\end{center}
\end{table}
Comparing Table \ref{t10} with Table \ref{t5}, we have $y_{1+112}=107$, $y_{1+212}=3$, $y_{2+112}=18$, $y_{2+212}=43$, $y_{++122}=19$ and $y_{++222}=8$ in addition to the counts in Table \ref{t9} of Example \ref{ex3}, which can be used to compute the MLE's of the various odds (see Section 3.2.2). From Table \ref{t10}, note that these MLE's satisfy $\frac{90}{2}\in\left(\frac{158}{68},\frac{1191}{8}\right)$, $\frac{1}{2}\not\in\left(\frac{7}{14},\frac{8}{2}\right)$, $\frac{90}{1}\in\left(\frac{158}{7},\frac{1191}{8}\right)$ and $\frac{2}{2}\not\in\left(\frac{8}{2},
\frac{68}{14}\right)$, while $\frac{107}{3}\in\left(\frac{8}{2},\frac{1191}{8}\right)$, $\frac{18}{43}\not\in\left(\frac{7}{14},\frac{158}{68}\right)$, $\frac{107}{18}\in\left(\frac{8}{7},\frac{1191}{158}\right)$ and $\frac{3}{43}\not\in\left(\frac{8}{68},
\frac{2}{14}\right)$. Hence, from Corollary \ref{cor3.2}, we deduce that the missing mechanism of $Y_{1}$ or $Y_{2}$ is most likely to be MAR and neither NMAR nor MCAR. Let $G^{2}$ denote the likelihood ratio statistic for testing the goodness of fit of the proposed model against the perfect fit model. When Models D1-D6 (see Section 3.2) are fitted to the data in Table \ref{t10}, the candidate models are those in which the missing mechanism is MAR for $Y_{2}$ (missing mechanism depends on $Y_{3}$) and that of $Y_{1}$ is MCAR or NMAR or MAR, based on $p$-values ($> 0.05$).  However, we deduce that the best fit model is NMAR for $Y_{1}$ and MAR for $Y_{2}$ (missing mechanism depends on $Y_{3}$) based on minimum $G^{2}$ value = 2.8076. This validates our earlier observation from Corollary \ref{cor3.2}. 

For assessing the accuracy of the proposed method, bootstrap resampling technique is used. We generate 10,000 random samples from the above models fitted to the data and check for those samples satisfying the conditions on $\widehat{\nu}$'s or on $\widehat{\omega}$'s in Corollary \ref{cor3.2}. The computed percentage of such samples for $\widehat{\nu}$'s ($\widehat{\omega}$'s) is 89.44 (93.41), 89.28 (93.75) and 89.93 (93.16) along with 96.56 (94.02) for the best fit model, which confirms the accuracy of the proposed method.
\end{example}

\begin{remark}\label{rem41}
For the two-way and three-way incomplete tables that we consider in Sections 2 and 3, the minimum no. of cells is 9 for a two-way $2\times 2\times 2\times 2$ table. Similarly, the no. of cells is 18 for a three-way $2\times 2\times 2\times 2\times 2$ table or 27 for a three-way $2\times 2\times 2\times 2\times 2\times 2$ table assuming each variable has only two levels. It is obviously more for larger tables with each variable having more categories. Under the assumption of strictly positive cell counts (which excludes sparse tables), the sample sizes (total cell counts) for incomplete tables will usually be large ($> 30$). So, the MLE's of the response and non-response odds used in our methods (see Corollaries \ref{cor2.1}, \ref{cor3.1} and \ref{cor3.2})  would be efficient and consistent. Thus our methods will perform well for most practical applications.
\end{remark}

\section{Conclusions}

In this paper, the missing data models for $I\times J\times 2\times 2,~ I\times J\times K\times 2,~ I\times J\times K\times 2\times 2$ and $I\times J\times K\times 2\times 2\times 2$ tables are introduced using hierarchical log-linear models. The forms of the various models are obtained by considering the missing mechanism of each variable and not the missingness of the outcome vector. Some particular properties of these missing data models are discussed in detail. We provide characteristic conditions for the various missing mechanisms of a variable in terms of response and non-response odds. These conditions help us to establish simple and useful procedures based only on the observed counts in the tables, which aid in the evaluation of MCAR, MAR and NMAR mechanisms of the variables. Finally, some real-life data analysis along with bootstrapping illustrates our results. Note that the models, techniques and results in this paper can be extended to higher dimensional incomplete tables also.

We now provide some comments on our methods of assessment for the MCAR or NMAR or MAR mechanism of each missing variable in various incomplete tables. Our aim has been to use estimates of the response and non-response odds involving only the fully and partially observed counts respectively in the tables, which are easy to calculate and simplify the verification process. If the missing mechanism of at least one of the variables is MCAR, then such models do not provide perfect fits for observed counts in the tables (see Ghosh and Vellaisamy (2016b)). Hence, estimates of the response and non-response odds are simple functions of the observed counts in this case.    

It is well known that the observed data are not sufficient to identify the mechanism
underlying missingness. So, the methods proposed in this paper are a form of sensitivity analysis to assess the MCAR, NMAR and MAR assumptions in an incomplete table. They are useful as data-analytic guidelines to perform model selection for the given incomplete data. An advantage of the proposed methods is that unlike existing selection procedures, there is no need to compute $p$-values, likelihood ratio statistics, AIC and BIC values to determine the missing data model. This is because these methods can suggest the missing mechanism of a variable directly (after checking some simple conditions from the given incomplete table) and hence do not constitute a goodness of fit testing procedure. Such methods work well when one of the variables is missing in an incomplete table (see Example \ref{ex3}). The best fit model is usually identified in this case. However, when two or more variables are missing in these tables, our methods can provide the probable missing mechanism (MAR or not) for each variable, but not the exact best fit model (see Examples \ref{ex1}, \ref{ex2} and \ref{ex4}).  Finally, the data analysis examples showed agreement and similar performance between the proposed methods and standard model selection criteria like $p$-values and $G^{2}$ in selecting plausible missing data models. Also, the results for bootstrapping procedures confirm the  accuracy of our proposed methods.

\section*{Appendix} 

\subsection*{A}{\it Proof of Theorem \ref{th2.2}}: First, we explore the conditions for which $\omega(i,i')\in OI(i,i')$ or $\omega(i,i')\not\in OI(i,i')$. \\
	1. Model M3 (NMAR for both $Y_{1}$ and $Y_{2}$) : \\
	Under Model M3, it can be shown that for any pair $(i,i')$ of $Y_{1}$, we have
	\begin{eqnarray*}
		\omega_{j}(i,i') &=& \exp\{\lambda_{Y_{1}}(i) - \lambda_{Y_{1}}(i') + \lambda_{Y_{1}Y_{2}}(i,j) - \lambda_{Y_{1}Y_{2}}(i',j) + \lambda_{Y_{1}R_{1}}(i,1) - \lambda_{Y_{1}R_{1}}(i',1)\}, \nonumber \\
		\omega(i,i') &=& \frac{\sum_{j}\exp\{\lambda_{Y_{1}}(i) + \lambda_{Y_{2}}(j) + \lambda_{Y_{1}Y_{2}}(i,j) + \lambda_{Y_{1}R_{1}}(i,1) + \lambda_{Y_{2}R_{2}}(j,2)\}}{\sum_{j}\exp\{\lambda_{Y_{1}}(i') + \lambda_{Y_{2}}(j) + \lambda_{Y_{1}Y_{2}}(i',j) + \lambda_{Y_{1}R_{1}}(i',1) + \lambda_{Y_{2}R_{2}}(j,2)\}},\nonumber \\
		\frac{\omega_{m}(i,i')}{\omega(i,i')} &=& \frac{\sum_{j}\exp\{\lambda_{Y_{2}}(j) + \lambda_{Y_{1}Y_{2}}(i',j) + \lambda_{Y_{1}Y_{2}}(i,m) + \lambda_{Y_{2}R_{2}}(j,2)\}}{\sum_{j}\exp\{\lambda_{Y_{2}}(j) + \lambda_{Y_{1}Y_{2}}(i,j) + \lambda_{Y_{1}Y_{2}}(i',m) + \lambda_{Y_{2}R_{2}}(j,2)\}}, \nonumber \\
		\frac{\omega_{n}(i,i')}{\omega(i,i')} &=& \frac{\sum_{j}\exp\{\lambda_{Y_{2}}(j) + \lambda_{Y_{1}Y_{2}}(i',j) + \lambda_{Y_{1}Y_{2}}(i,n) + \lambda_{Y_{2}R_{2}}(j,2)\}}{\sum_{j}\exp\{\lambda_{Y_{2}}(j) + \lambda_{Y_{1}Y_{2}}(i,j) + \lambda_{Y_{1}Y_{2}}(i',n) + \lambda_{Y_{2}R_{2}}(j,2)\}}. \nonumber   
	\end{eqnarray*}
	Now
	\begin{eqnarray*}
		&&\omega_{j}(i,i') < \omega_{m}(i,i') \nonumber \\
		&\Rightarrow& 1 < \frac{\sum_{j}\exp\{\lambda_{Y_{2}}(j) + \lambda_{Y_{1}Y_{2}}(i',j) + \lambda_{Y_{1}Y_{2}}(i,m) + \lambda_{Y_{2}R_{2}}(j,2)\}}{\sum_{j}\exp\{\lambda_{Y_{2}}(j) + \lambda_{Y_{1}Y_{2}}(i,j) + \lambda_{Y_{1}Y_{2}}(i',m) + \lambda_{Y_{2}R_{2}}(j,2)\}} \nonumber \\
		&\Rightarrow& \frac{\omega_{m}(i,i')}{\omega(i,i')} > 1. \nonumber 
	\end{eqnarray*}
	Also,
	\begin{eqnarray*}
		&&\omega_{j}(i,i') > \omega_{n}(i,i') \nonumber \\
		&\Rightarrow& 1 > \frac{\sum_{j}\exp\{\lambda_{Y_{2}}(j) + \lambda_{Y_{1}Y_{2}}(i',j) + \lambda_{Y_{1}Y_{2}}(i,n) + \lambda_{Y_{2}R_{2}}(j,2)\}}{\sum_{j}\exp\{\lambda_{Y_{2}}(j) + \lambda_{Y_{1}Y_{2}}(i,j) + \lambda_{Y_{1}Y_{2}}(i',n) + \lambda_{Y_{2}R_{2}}(j,2)\}} \nonumber \\
		&\Rightarrow& \frac{\omega_{n}(i,i')}{\omega(i,i')} < 1. \nonumber
	\end{eqnarray*}
	Hence, under Model 3, $\omega(i,i')\in(\omega_{n}(i,i'),~\omega_{m}(i,i')) = OI(i,i')$ for any pair $(i,i')$ of $Y_{1}$ if $|\lambda_{Y_{2}R_{2}}(j,2)| < \infty$. 
	\vone\noindent
	2. Model M5 (MAR for both $Y_{1}$ and $Y_{2}$) : \\
	Under Model M5, it can be shown that for any pair $(i,i')$ of $Y_{1}$, we have
	\begin{eqnarray*}
		\omega_{j}(i,i') &=& \exp\{\lambda_{Y_{1}}(i) - \lambda_{Y_{1}}(i') + \lambda_{Y_{1}Y_{2}}(i,j) - \lambda_{Y_{1}Y_{2}}(i',j) + \lambda_{Y_{1}R_{2}}(i,1) - \lambda_{Y_{1}R_{2}}(i',1)\},\nonumber \\
		\omega(i,i') &=& \frac{\sum_{j}\exp\{\lambda_{Y_{1}}(i) + \lambda_{Y_{2}}(j) + \lambda_{Y_{1}Y_{2}}(i,j) + \lambda_{Y_{1}R_{2}}(i,2) + \lambda_{Y_{2}R_{1}}(j,1)\}}{\sum_{j}\exp\{\lambda_{Y_{1}}(i') + \lambda_{Y_{2}}(j) + \lambda_{Y_{1}Y_{2}}(i',j) + \lambda_{Y_{1}R_{2}}(i',2) + \lambda_{Y_{2}R_{1}}(j,1)\}},\nonumber \\
		\frac{\omega_{m}(i,i')}{\omega(i,i')} &=& \exp\{2(\lambda_{Y_{1}R_{2}}(i',2) - \lambda_{Y_{1}R_{2}}(i,2))\}\times B_{m}(i,i'),\nonumber \\
		\frac{\omega_{n}(i,i')}{\omega(i,i')} &=& \exp\{2(\lambda_{Y_{1}R_{2}}(i',2) - \lambda_{Y_{1}R_{2}}(i,2))\}\times B_{n}(i,i'),\nonumber    
	\end{eqnarray*} 
	Now
	\begin{eqnarray*}
		&&\omega_{j}(i,i') < \omega_{m}(i,i') \nonumber \\
		&\Rightarrow& 1 < \frac{\sum_{j}\exp\{\lambda_{Y_{2}}(j) + \lambda_{Y_{1}Y_{2}}(i',j) + \lambda_{Y_{1}Y_{2}}(i,m) + \lambda_{Y_{2}R_{1}}(j,1)\}}{\sum_{j}\exp\{\lambda_{Y_{2}}(j) + \lambda_{Y_{1}Y_{2}}(i,j) + \lambda_{Y_{1}Y_{2}}(i',m) + \lambda_{Y_{2}R_{1}}(j,1)\}} \nonumber \\
		&\Rightarrow& B_{m}(i,i') > 1. \nonumber 
	\end{eqnarray*}
	Also,
	\begin{eqnarray*}
		&&\omega_{j}(i,i') > \omega_{n}(i,i') \nonumber \\
		&\Rightarrow& 1 > \frac{\sum_{j}\exp\{\lambda_{Y_{2}}(j) + \lambda_{Y_{1}Y_{2}}(i',j) + \lambda_{Y_{1}Y_{2}}(i,n) + \lambda_{Y_{2}R_{1}}(j,1)\}}{\sum_{j}\exp\{\lambda_{Y_{2}}(j) + \lambda_{Y_{1}Y_{2}}(i,j) + \lambda_{Y_{1}Y_{2}}(i',n) + \lambda_{Y_{2}R_{1}}(j,1)\}} \nonumber \\
		&\Rightarrow& B_{n}(i,i') < 1. \nonumber
	\end{eqnarray*}
	Suppose $\omega(i,i')\in OI(i,i')\Leftrightarrow \frac{\omega_{m}(i,i')}{\omega(i,i')} > 1$ and $\frac{\omega_{n}(i,i')}{\omega(i,i')} < 1$. Then
	\begin{equation*}
	\frac{\omega_{m}(i,i')}{\omega(i,i')} > 1 \Leftrightarrow \lambda_{Y_{1}R_{2}}(i',2) - \lambda_{Y_{1}R_{2}}(i,2) > -\frac{1}{2}\log B_{m}(i,i').
	\end{equation*} 
	Also,
	\begin{equation*}
	\frac{\omega_{n}(i,i')}{\omega(i,i')} < 1 \Leftrightarrow \lambda_{Y_{1}R_{2}}(i',2) - \lambda_{Y_{1}R_{2}}(i,2) < -\frac{1}{2}\log B_{n}(i,i').
	\end{equation*}
	Hence $\omega(i,i')\in OI(i,i')$ iff $-\frac{1}{2}\log B_{m}(i,i') < \lambda_{Y_{1}R_{2}}(i',2) - \lambda_{Y_{1}R_{2}}(i,2) < -\frac{1}{2}\log B_{n}(i,i')$. Equivalently, $\omega(i,i')\not\in OI(i,i')$ iff $\lambda_{Y_{1}R_{2}}(i',2) - \lambda_{Y_{1}R_{2}}(i,2) > -\frac{1}{2}\log B_{n}(i,i')$ or $\lambda_{Y_{1}R_{2}}(i',2) - \lambda_{Y_{1}R_{2}}(i,2)  < -\frac{1}{2}\log B_{m}(i,i')$. Thus under Model M5, only one of Conditions 1 and 2 holds:
	\begin{enumerate}
		\item[1.] $\omega(i,i')\in OI(i,i')$ iff $-\frac{1}{2}\log B_{m}(i,i') < \lambda_{Y_{1}R_{2}}(i',2) - \lambda_{Y_{1}R_{2}}(i,2) < -\frac{1}{2}\log B_{n}(i,i'),$ 
		\item[2.] $\omega(i,i')\not\in OI(i,i')$ iff $\lambda_{Y_{1}R_{2}}(i',2) - \lambda_{Y_{1}R_{2}}(i,2) > -\frac{1}{2}\log B_{n}(i,i')$ or $\lambda_{Y_{1}R_{2}}(i',2) - \lambda_{Y_{1}R_{2}}(i,2)  < -\frac{1}{2}\log B_{m}(i,i')$. 
	\end{enumerate}
	\vone\noindent
	Similar conditions can be obtained under the other models. Let $B'_{m}(i,i')$ and $B'_{n}(i,i')$ denote $B_{m}(i,i')$ and $B_{n}(i,i')$ respectively with $\lambda_{Y_{2}R_{1}}(j,1)=0$. Then Table \ref{t2} summarizes the conditions under which $\omega(i,i')\in OI(i,i')$ for any pair $(i,i')$ of $Y_{1}$ for Models M1-M9. 
	\begin{table}[ht]
		\caption{Conditions for $\omega(i,i')\in OI(i,i')$ under missing data models in an $I\times J\times 2\times 2$ Incomplete Table.}\label{t2}
		\begin{center}
			$
			{\renewcommand{\arraystretch}{1.3}
				\begin{array}{|c|c|}\hline
				\textrm{Model} & \textrm{Conditions} \\ \hline
				\textrm{Model M1} & \textrm{Nil} \\ \hline
				\textrm{Model M2} & -\frac{1}{2}\log B'_{m}(i,i') < \lambda_{Y_{1}R_{2}}(i',2) - \lambda_{Y_{1}R_{2}}(i,2) < -\frac{1}{2}\log B'_{n}(i,i') \\ \hline
				\textrm{Model M3} & |\lambda_{Y_{2}R_{2}}(j,2)| < \infty \\ \hline
				\textrm{Model M4} & \textrm{Nil} \\ \hline
				\textrm{Model M5} & -\frac{1}{2}\log B_{m}(i,i') < \lambda_{Y_{1}R_{2}}(i',2) - \lambda_{Y_{1}R_{2}}(i,2) < -\frac{1}{2}\log B_{n}(i,i') \\ \hline
				\textrm{Model M6} & |\lambda_{Y_{2}R_{2}}(j,2)| < \infty \\ \hline
				\textrm{Model M7} & -\frac{1}{2}\log B'_{m}(i,i') < \lambda_{Y_{1}R_{2}}(i',2) - \lambda_{Y_{1}R_{2}}(i,2) < -\frac{1}{2}\log B'_{n}(i,i') \\ \hline
				\textrm{Model M8} & |\lambda_{Y_{2}R_{2}}(j,2)| < \infty \\ \hline
				\textrm{Model M9} & \textrm{Nil} \\ \hline
				\end{array}}
			$
		\end{center} 
	\end{table}
	The proof of Part (a) follows from the Conditions in Table \ref{t2} for which $\omega(i,i')\in OI(i,i')$ under Models M1, M3, M4, M6, M8 and M9. Also, the proof of Part (b) follows from the Conditions in Table \ref{t2} for which $\omega(i,i')\in OI(i,i')$ or $\omega(i,i')\not\in OI(i,i')$ under Models M2, M5 and M7.

\subsection*{B}{\it Proof of Theorem \ref{th2.1}}:	Similar to the proof of Theorem \ref{th2.2}, we study the behaviour of the relevant odds. \\
	1. Model M3 (NMAR for both $Y_{1}$ and $Y_{2}$) : \\
	Consider the response and non-response odds based on $\pi$ for any pair $(j,j')$ of $Y_{2}$. Then using similar arguments as in the proof of Theorem \ref{th2.2}, it can be shown that $\nu(j,j')\in(\nu_{n}(j,j'),~\nu_{m}(j,j')) = OI(j,j')$ for any pair $(i,i')$ of $Y_{1}$ if $|\lambda_{Y_{1}R_{1}}(i,2)| < \infty$. 
	\vone\noindent
	2. Model M5 (MAR for both $Y_{1}$ and $Y_{2}$) : \\
	Consider the response and non-response odds based on $\pi$ for any pair $(j,j')$ of $Y_{2}$. Then using similar arguments as in the proof of Theorem \ref{th2.2}, it can be shown that 
	\begin{eqnarray*}
		\frac{\nu_{m}(j,j')}{\nu(j,j')} &=& \exp\{2(\lambda_{Y_{2}R_{1}}(j',2) - \lambda_{Y_{2}R_{1}}(j,2))\}\times A_{m}(j,j'),\nonumber \\
		\frac{\nu_{n}(j,j')}{\nu(j,j')} &=& \exp\{2(\lambda_{Y_{2}R_{1}}(j',2) - \lambda_{Y_{2}R_{1}}(j,2))\}\times A_{n}(j,j').\nonumber
	\end{eqnarray*}
	Now $\nu_{i}(j,j') < \nu_{m}(j,j')\Rightarrow A_{m}(j,j') > 1$ and $\nu_{i}(j,j') > \nu_{n}(j,j')\Rightarrow A_{n}(j,j') < 1$. Suppose $\nu(j,j')\in OI(j,j')\Leftrightarrow \frac{\nu_{m}(j,j')}{\nu(j,j')} > 1$ and $\frac{\nu_{n}(j,j')}{\nu(j,j')} < 1$. Then it can be shown that $\nu(j,j')\in OI(j,j')$ iff $-\frac{1}{2}\log A_{m}(j,j') < \lambda_{Y_{2}R_{1}}(j',2) - \lambda_{Y_{2}R_{1}}(j,2) < -\frac{1}{2}\log A_{n}(j,j')$. Equivalently, $\nu(j,j')\not\in OI(j,j')$ iff $\lambda_{Y_{2}R_{1}}(j',2) - \lambda_{Y_{2}R_{1}}(j,2) > -\frac{1}{2}\log A_{n}(j,j')$ or $\lambda_{Y_{2}R_{1}}(j',2) - \lambda_{Y_{2}R_{1}}(j,2)  < -\frac{1}{2}\log A_{m}(j,j')$. Thus under Model M5, only one of Conditions 1 and 2 holds:
	\begin{enumerate}
		\item[1.] $\nu(j,j')\in OI(j,j')$ iff $-\frac{1}{2}\log A_{m}(j,j') < \lambda_{Y_{2}R_{1}}(j',2) - \lambda_{Y_{2}R_{1}}(j,2) < -\frac{1}{2}\log A_{n}(j,j')$, 
		\item[2.] $\nu(j,j')\not\in OI(j,j')$ iff $\lambda_{Y_{2}R_{1}}(j',2) - \lambda_{Y_{2}R_{1}}(j,2) > -\frac{1}{2}\log A_{n}(j,j')$ or $\lambda_{Y_{2}R_{1}}(j',2) - \lambda_{Y_{2}R_{1}}(j,2)  < -\frac{1}{2}\log A_{m}(j,j')$.
	\end{enumerate}
	\vone\noindent
	Similar conditions can be obtained under the other models. Let $A'_{m}(j,j')$ and $A'_{n}(j,j')$ denote $A_{m}(j,j')$ and $A_{n}(j,j')$ respectively with $\lambda_{Y_{1}R_{2}}(i,1)=0$. Then Table \ref{t3} summarizes the conditions under which $\nu(j,j')\in OI(j,j')$ for any pair $(j,j')$ of $Y_{2}$ for Models M1-M9. 
	\begin{table}[ht]
		\caption{Conditions for $\nu(j,j')\in OI(j,j')$ under missing data models in an $I\times J\times 2\times 2$ Incomplete Table}\label{t3}
		\begin{center}
			$
			{\renewcommand{\arraystretch}{1.3}
				\begin{array}{|c|c|}\hline
				\textrm{Model} & \textrm{Conditions} \\ \hline
				\textrm{Model M1} & |\lambda_{Y_{1}R_{1}}(i,2)| < \infty \\ \hline
				\textrm{Model M2} & |\lambda_{Y_{1}R_{1}}(i,2)| < \infty \\ \hline
				\textrm{Model M3} & |\lambda_{Y_{1}R_{1}}(i,2)| < \infty \\ \hline
				\textrm{Model M4} & -\frac{1}{2}\log A'_{m}(j,j') < \lambda_{Y_{2}R_{1}}(j',2) - \lambda_{Y_{2}R_{1}}(j,2) < -\frac{1}{2}\log A'_{n}(j,j') \\ \hline
				\textrm{Model M5} & -\frac{1}{2}\log A_{m}(j,j') < \lambda_{Y_{2}R_{1}}(j',2) - \lambda_{Y_{2}R_{1}}(j,2) < -\frac{1}{2}\log A_{n}(j,j') \\ \hline
				\textrm{Model M6} & -\frac{1}{2}\log A'_{m}(j,j') < \lambda_{Y_{2}R_{1}}(j',2) - \lambda_{Y_{2}R_{1}}(j,2) < -\frac{1}{2}\log A'_{n}(j,j') \\ \hline
				\textrm{Model M7} & \textrm{Nil} \\ \hline
				\textrm{Model M8} & \textrm{Nil} \\ \hline
				\textrm{Model M9} & \textrm{Nil} \\ \hline
				\end{array}}
			$
		\end{center}
	\end{table} 
	The proof of Part (a) now follows from the Conditions in Table \ref{t3} for which $\nu(j,j')\in OI(j,j')$ under Models M1-M3 and M7-M9. Also, the proof of Part (b) follows from Conditions in Table \ref{t3} for which $\nu(j,j')\in OI(j,j')$ or $\nu(j,j')\not\in OI(j,j')$ under Models M4-M6.

\subsection*{C}{\it Proof of Theorem \ref{th3.1}}: Consider the models C1 and C4 for which the missing mechanism of $Y_{1}$ is NMAR and MCAR respectively in an $I\times J\times K\times 2$ table. Under Model C1, we have 
	\begin{eqnarray*}
		\nu_{ik}(j,j') &=& \exp\{\lambda_{Y_{2}}(j) - \lambda_{Y_{2}}(j') + \lambda_{Y_{1}Y_{2}}(i,j) - \lambda_{Y_{1}Y_{2}}(i,j') + \lambda_{Y_{2}Y_{3}}(j,k) - \lambda_{Y_{2}Y_{3}}(j',k)\}, \nonumber \\
		\nu_{k}(j,j') &=& \frac{\sum_{i}\exp\{\lambda_{Y_{1}}(i) + \lambda_{Y_{2}}(j) + \lambda_{Y_{1}Y_{2}}(i,j) + \lambda_{Y_{1}Y_{3}}(i,k) + \lambda_{Y_{2}Y_{3}}(j,k) + \lambda_{Y_{1}R}(i,2)\}}{\sum_{i}\exp\{\lambda_{Y_{1}}(i) + \lambda_{Y_{2}}(j') + \lambda_{Y_{1}Y_{2}}(i,j') + \lambda_{Y_{1}Y_{3}}(i,k) + \lambda_{Y_{2}Y_{3}}(j',k) + \lambda_{Y_{1}R}(i,2)\}},\nonumber \\
		\frac{\nu_{mk}(j,j')}{\nu_{k}(j,j')} &=& \frac{\sum_{i}\exp\{\lambda_{Y_{1}}(i) + \lambda_{Y_{1}Y_{2}}(i,j') + \lambda_{Y_{1}Y_{2}}(m,j) + \lambda_{Y_{1}Y_{3}}(i,k) + \lambda_{Y_{1}R}(i,2)\}}{\sum_{i}\exp\{\lambda_{Y_{1}}(i) + \lambda_{Y_{1}Y_{2}}(i,j) + \lambda_{Y_{1}Y_{2}}(m,j') + \lambda_{Y_{1}Y_{3}}(i,k) + \lambda_{Y_{1}R}(i,2)\}},\nonumber \\
		\frac{\nu_{nk}(j,j')}{\nu_{k}(j,j')} &=& \frac{\sum_{i}\exp\{\lambda_{Y_{1}}(i) + \lambda_{Y_{1}Y_{2}}(i,j') + \lambda_{Y_{1}Y_{2}}(n,j) + \lambda_{Y_{1}Y_{3}}(i,k) + \lambda_{Y_{1}R}(i,2)\}}{\sum_{i}\exp\{\lambda_{Y_{1}}(i) + \lambda_{Y_{1}Y_{2}}(i,j) + \lambda_{Y_{1}Y_{2}}(n,j') + \lambda_{Y_{1}Y_{3}}(i,k) + \lambda_{Y_{1}R}(i,2)\}}.\nonumber
	\end{eqnarray*}
	Now
	\begin{eqnarray*}
		&& \nu_{ik}(j,j') < \nu_{mk}(j,j') \nonumber \\
		&\Rightarrow& 1 < \frac{\sum_{i}\exp\{\lambda_{Y_{1}}(i) + \lambda_{Y_{1}Y_{2}}(i,j') + \lambda_{Y_{1}Y_{2}}(m,j) + \lambda_{Y_{1}Y_{3}}(i,k) + \lambda_{Y_{1}R}(i,2)\}}{\sum_{i}\exp\{\lambda_{Y_{1}}(i) + \lambda_{Y_{1}Y_{2}}(i,j) + \lambda_{Y_{1}Y_{2}}(m,j') + \lambda_{Y_{1}Y_{3}}(i,k) + \lambda_{Y_{1}R}(i,2)\}}\nonumber \\
		&\Rightarrow& \nu_{mk}(j,j') > \nu_{k}(j,j'). \nonumber  
	\end{eqnarray*}
	Also,
	\begin{eqnarray*}
		&& \nu_{ik}(j,j') > \nu_{nk}(j,j') \nonumber \\
		&\Rightarrow & 1 > \frac{\sum_{i}\exp\{\lambda_{Y_{1}}(i) + \lambda_{Y_{1}Y_{2}}(i,j') + \lambda_{Y_{1}Y_{2}}(n,j) + \lambda_{Y_{1}Y_{3}}(i,k) + \lambda_{Y_{1}R}(i,2)\}}{\sum_{i}\exp\{\lambda_{Y_{1}}(i) + \lambda_{Y_{1}Y_{2}}(i,j) + \lambda_{Y_{1}Y_{2}}(n,j') + \lambda_{Y_{1}Y_{3}}(i,k) + \lambda_{Y_{1}R}(i,2)\}}\nonumber \\
		&\Rightarrow & \nu_{nk}(j,j') < \nu_{k}(j,j'). \nonumber  
	\end{eqnarray*}
	Hence $\nu_{k}(j,j')\in(\nu_{nk}(j,j'),\nu_{mk}(j,j')) = OI_{k}(j,j')$ if $|\lambda_{Y_{1}R}(i,2)| < \infty$. Using similar arguments, we can show that $\nu_{j}(k,k')\in(\nu_{nj}(k,k'),\nu_{mj}(k,k')) = OI_{j}(k,k')$ if $|\lambda_{Y_{1}R}(i,2)| < \infty$. Thus under Model C1, $\nu_{k}(j,j')\in OI_{k}(j,j')$ and $\nu_{j}(k,k')\in OI_{j}(k,k')$ if $|\lambda_{Y_{1}R}(i,2)| < \infty$. 
	Using similar arguments, it can also be shown that under Model C4 (MCAR for $Y_{1}$), both $\nu_{k}(j,j')\in OI_{k}(j,j')$ and $\nu_{j}(k,k')\in OI_{j}(k,k')$ hold.
	This completes the proof.

\subsection*{D}{\it Proof of Theorem \ref{th3.2}}: Consider models C2 and C3 for which the missing mechanism of $Y_{1}$ is MAR in an $I\times J\times K\times 2$ table. Under Model C2, we have
	\begin{eqnarray*}
		\nu_{ik}(j,j') &=& \exp\{\lambda_{Y_{2}}(j) - \lambda_{Y_{2}}(j') + \lambda_{Y_{1}Y_{2}}(i,j) - \lambda_{Y_{1}Y_{2}}(i,j') + \lambda_{Y_{2}Y_{3}}(j,k) - \lambda_{Y_{2}Y_{3}}(j',k) \nonumber \\
		&& + \lambda_{Y_{2}R}(j,1) - \lambda_{Y_{2}R}(j',1)\},\nonumber \\
		\nu_{k}(j,j') &=& \frac{\sum_{i}\exp\{\lambda_{Y_{1}}(i) + \lambda_{Y_{2}}(j) + \lambda_{Y_{1}Y_{2}}(i,j) + \lambda_{Y_{1}Y_{3}}(i,k) + \lambda_{Y_{2}Y_{3}}(j,k) + \lambda_{Y_{2}R}(j,2)\}}{\sum_{i}\exp\{\lambda_{Y_{1}}(i) + \lambda_{Y_{2}}(j') + \lambda_{Y_{1}Y_{2}}(i,j') + \lambda_{Y_{1}Y_{3}}(i,k) + \lambda_{Y_{2}Y_{3}}(j',k) + \lambda_{Y_{2}R}(j',2)\}},\nonumber \\
		\frac{\nu_{mk}(j,j')}{\nu_{k}(j,j')} &=& \exp[2\{\lambda_{Y_{2}R}(j',2) - \lambda_{Y_{2}R}(j,2)\}]\times A_{mk}(j,j'), \nonumber \\
		\frac{\nu_{nk}(j,j')}{\nu_{k}(j,j')} &=& \exp[2\{\lambda_{Y_{2}R}(j',2) - \lambda_{Y_{2}R}(j,2)\}]\times A_{nk}(j,j'). \nonumber
	\end{eqnarray*}
	Now $\nu_{ik}(j,j') < \nu_{mk}(j,j')\Rightarrow A_{mk}(j,j') > 1$ and $\nu_{ik}(j,j') > \nu_{nk}(j,j')\Rightarrow A_{nk}(j,j') < 1$. Suppose $\nu_{k}(j,j')\in OI_{k}(j,j')\Leftrightarrow\frac{\nu_{mk}(j,j')}{\nu_{k}(j,j')} > 1$ and $\frac{\nu_{nk}(j,j')}{\nu_{k}(j,j')} < 1$. Then we have
	\begin{equation*}
	-\frac{1}{2}\log A_{mk}(j,j') < \lambda_{Y_{2}R}(j',2) - \lambda_{Y_{2}R}(j,2) < -\frac{1}{2}\log A_{nk}(j,j').
	\end{equation*}
	Equivalently, $\nu_{k}(j,j')\not\in OI_{k}(j,j')\Leftrightarrow\nu_{k}(j,j')\not\in OI_{k}(j,j')$ iff $\lambda_{Y_{2}R}(j',2) - \lambda_{Y_{2}R}(j,2) > -\frac{1}{2}\log A_{nk}(j,j')$ or $\lambda_{Y_{2}R}(j',2) - \lambda_{Y_{2}R}(j,2)  < -\frac{1}{2}\log A_{mk}(j,j')$.
	
	Next, consider the response and non-response odds based on $\pi$ for any pair $(k,k')$ of $Y_{3}$ and $1\leq j\leq J$. Then
	\begin{eqnarray*}
		\nu_{ij}(k,k') &=& \exp\{\lambda_{Y_{3}}(k) - \lambda_{Y_{3}}(k') + \lambda_{Y_{1}Y_{3}}(i,k) - \lambda_{Y_{1}Y_{3}}(i,k') + \lambda_{Y_{2}Y_{3}}(j,k) - \lambda_{Y_{2}Y_{3}}(j,k')\}, \nonumber \\
		\nu_{j}(k,k') &=& \frac{\sum_{i}\exp\{\lambda_{Y_{1}}(i) + \lambda_{Y_{3}}(k) + \lambda_{Y_{1}Y_{2}}(i,j) + \lambda_{Y_{1}Y_{3}}(i,k) + \lambda_{Y_{2}Y_{3}}(j,k)\}}{\sum_{i}\exp\{\lambda_{Y_{1}}(i) + \lambda_{Y_{3}}(k') + \lambda_{Y_{1}Y_{2}}(i,j) + \lambda_{Y_{1}Y_{3}}(i,k') + \lambda_{Y_{2}Y_{3}}(j,k')\}},\nonumber \\
		\frac{\nu_{mj}(k,k')}{\nu_{j}(k,k')} &=& \frac{\sum_{i}\exp\{\lambda_{Y_{1}}(i) + \lambda_{Y_{1}Y_{2}}(i,j) + \lambda_{Y_{1}Y_{3}}(i,k') + \lambda_{Y_{1}Y_{3}}(m,k)\}}{\sum_{i}\exp\{\lambda_{Y_{1}}(i) + \lambda_{Y_{1}Y_{2}}(i,j) + \lambda_{Y_{1}Y_{3}}(i,k) + \lambda_{Y_{1}Y_{3}}(m,k')\}}. \nonumber 
	\end{eqnarray*}
	Now it can be shown that $\nu_{ij}(k,k') < \nu_{mj}(k,k')\Rightarrow\nu_{mj}(k,k') > \nu_{j}(k,k')$ and $\nu_{ij}(k,k') > \nu_{nj}(k,k')\Rightarrow\nu_{nj}(k,k') < \nu_{j}(k,k')$. Hence $\nu_{j}(k,k')\in(\nu_{nj}(k,k'), \nu_{mj}(k,k')) = OI_{j}(k,k')$. So under Model C2, both the following Conditions 1a and 1b hold: 
	\begin{enumerate}
		\item[1a.] Only one of the following conditions holds. \\
		(i) $\nu_{k}(j,j')\in OI_{k}(j,j')$ iff $-\frac{1}{2}\log A_{mk}(j,j') < \lambda_{Y_{2}R}(j',2) - \lambda_{Y_{2}R}(j,2) < -\frac{1}{2}\log A_{nk}(j,j')$, \\
		(ii) $\nu_{k}(j,j')\not\in OI_{k}(j,j')$ iff $\lambda_{Y_{2}R}(j',2) - \lambda_{Y_{2}R}(j,2) > -\frac{1}{2}\log A_{nk}(j,j')$ or $\lambda_{Y_{2}R}(j',2) - \lambda_{Y_{2}R}(j,2)  < -\frac{1}{2}\log A_{mk}(j,j')$ 
		\item[1b.] $\nu_{j}(k,k')\in OI_{j}(k,k')$.
	\end{enumerate}
	Using similar arguments, we can show that under Model C3 (MAR for $Y_{1}$), both the following Conditions 2a and 2b hold:
	\begin{enumerate}
		\item[2a.] $\nu_{k}(j,j')\in OI_{k}(j,j')$, 
		\item[2b.] Only one of the conditions below holds: \\
		(i) $\nu_{j}(k,k')\in OI_{j}(k,k')$ iff $-\frac{1}{2}\log A_{mj}(k,k') < \lambda_{Y_{3}R}(k',2) - \lambda_{Y_{3}R}(k,2) < -\frac{1}{2}\log A_{nj}(k,k'),$ \\
		(ii) $\nu_{j}(k,k')\not\in OI_{j}(k,k')$ iff $\lambda_{Y_{3}R}(k',2) - \lambda_{Y_{3}R}(k,2) > -\frac{1}{2}\log A_{nj}(k,k')$ or $\lambda_{Y_{3}R}(k',2) - \lambda_{Y_{3}R}(k,2)  < -\frac{1}{2}\log A_{mj}(k,k')$.
	\end{enumerate}
	By assumption, the MAR mechanism of $Y_{1}$ can depend on $Y_{2}$ or $Y_{3}$ but not both. Hence, only one of Conditions (1a,1b) and (2a,2b) characterizes the MAR mechanism of $Y_{1}$.

\subsection*{E}{\it Proof of Theorem \ref{th4.1}}: The proof is similar to that of Theorems \ref{th3.1} and \ref{th3.2}.

\vone\noindent
{\bf Acknowledgements:} The authors are grateful to both the reviewers for going through the manuscript carefully and offering various comments and suggestions, which greatly improved the quality of the paper.

\end{document}